\documentclass[runningheads]{llncs}
\usepackage[T1]{fontenc}
\usepackage{graphicx}
\usepackage[inline]{enumitem}
\usepackage{mathtools}
\usepackage{amsmath}
\usepackage{multirow}
\usepackage{amssymb}
\usepackage{hyperref} 
\usepackage{booktabs}
\usepackage{xcolor}

\def \PMp{\mathcal{PM}_{\vec v}}
\def \PP{\mathcal{P}}
\def \EE{\mathbb{E}}
\def \erre{\mathbb{R}}
\def \enne{\mathbb{N}}

\def \SCo{SC}

\DeclarePairedDelimiter\floor{\lfloor}{\rfloor}
\begin{document}
\title{The \texorpdfstring{$k$}{k}-Facility Location Problem Via Optimal Transport: A Bayesian Study of the Percentile Mechanisms}
\titlerunning{The \texorpdfstring{$k$}{k}-Facility Location Problem Via Optimal Transport}
% If the paper title is too long for the running head, you can set
% an abbreviated paper title here
%
\author{Gennaro Auricchio\inst{1}\orcidID{0000-0002-4285-8887} \and
Jie Zhang \inst{2}}
\authorrunning{G. Auricchio and J. Zhang}
% First names are abbreviated in the running head.
% If there are more than two authors, 'et al.' is used.
%
\institute{Università degli Studi di Padova, Padua, Italy  \and
University of Bath, Bath, United Kingdom\\ \email{gennaro.auricchio@unipd.it,jz2558@bath.ac.uk}}
% First names are abbreviated in the running head.
% If there are more than two authors, 'et al.' is used.
%
% \institute{Princeton University, Princeton NJ 08544, USA \and
% Springer Heidelberg, Tiergartenstr. 17, 69121 Heidelberg, Germany
% \email{lncs@springer.com}\\
% \url{http://www.springer.com/gp/computer-science/lncs} \and
% ABC Institute, Rupert-Karls-University Heidelberg, Heidelberg, Germany\\
% \email{\{abc,lncs\}@uni-heidelberg.de}}
%
\maketitle              % typeset the header of the contribution
\begin{abstract}
In this paper, we investigate the $k$-Facility Location Problem ($k$-FLP) within the Bayesian Mechanism Design framework, in which agents' preferences are samples of a probability distributed on a line. 
Our primary contribution is characterising the asymptotic behavior of percentile mechanisms, which varies according to the distribution governing the agents' types.
To achieve this, we connect the $k$-FLP and projection problems in the Wasserstein space. 
Owing to this relation, we show that the ratio between the expected cost of a percentile mechanism and the expected optimal cost is asymptotically bounded. 
Furthermore, we characterize the limit of this ratio and analyze its convergence speed. 
Our asymptotic study is complemented by deriving an upper bound on the Bayesian approximation ratio, applicable when the number of agents $n$ exceeds the number of facilities $k$.
We also characterize the optimal percentile mechanism for a given agent's distribution through a system of $k$ equations.
Finally, we estimate the optimality loss incurred when the optimal percentile mechanism is derived using an approximation of the agents' distribution rather than the actual distribution.

\keywords{Bayesian Mechanism Design  \and Facility Location Problem \and Optimal Transport.}
\end{abstract}

\section{Introduction}
The scope of Mechanism Design is defining procedures that aggregate a group of agents' private information for optimizing a social objective.
Nevertheless, merely optimizing the social objective based on the reported preferences often leads to undesired manipulation due to the agents' self-interested behaviour.
For this reason, one of the most important properties a mechanism should possess is \textit{truthfulness}, which guarantees that no agent benefits from misreporting its private information. 
This stringent property is often incompatible with the optimization of the social objective, so we have to compromise on a sub-optimal solution.
To quantify the efficiency loss, Nisan and Ronen introduced the notion of \textit{approximation ratio}, which is the highest ratio between the social objective achieved by a truthful mechanism and the optimal social objective achievable over all the possible agents' reports \cite{nisan1999algorithmic}. 
One of the most famous examples of these problems is the $k$-Facility Location Problem ($k$-FLP), where a central authority has to locate $k$ facilities amongst $n$ self-interested agents.
Every agent needs to access the facility, so they would prefer to have one of the facilities placed as close as possible to their position.
Despite its simplicity, this problem and its variants have found a wide range of applications in fields such as disaster relief \cite{doi:10.1080/13675560701561789}, supply chain management \cite{MELO2009401}, healthcare \cite{ahmadi2017survey}, and public facilities accessibility \cite{barda1990multicriteria}.
The study of the $k$-FLP from an algorithmic mechanism design viewpoint was initiated by Procaccia and Tennenholtz.
In their seminal work \cite{procaccia2013approximate}, they considered the problem of locating one facility amongst a group of agents situated in a line.
They were the first to design an allocation process that places the facility while keeping in mind that every agent is self-interested, i.e. that agents would manipulate the process in its favour if able. 
Subsequently, a variety of methods with fixed approximation ratios for positioning one or two facilities on different types of structures such as double-peaked \cite{DBLP:journals/aamas/Filos-RatsikasL17}, trees \cite{DBLP:conf/atal/FilimonovM21}, circles \cite{DBLP:conf/sigecom/LuSWZ10,DBLP:conf/wine/LuWZ09}, graphs \cite{10.2307/40800845,DBLP:conf/sigecom/DokowFMN12}, and metric spaces \cite{DBLP:conf/sagt/Meir19,DBLP:conf/sigecom/TangYZ20} were introduced.
These positive outcomes, however, pertain to scenarios with a limited number of agents or when the facilities to place are at most $2$.
The approximation ratio results are much more negative when we move to three or more facilities.
Fotakis and Tzamos \cite{fotakis2014power} showed that for every $k \ge 3$, there does not exist any deterministic, anonymous, and truthful mechanisms with a bounded approximation ratio for the $k$-FLP on the line, even for instances with $k+2$ agents.
Nonetheless, it is possible to define truthful mechanisms with bounded approximation ratio when the number of agents is equal to the number of facilities plus one \cite{escoffier2011strategy} or by considering randomized mechanisms \cite{fotakis2013strategyproof}.
Our study concerns a class of truthful mechanisms for the generic $k$-FLP, the \textit{percentile mechanisms}, introduced in \cite{sui2013analysis}.
Although every percentile mechanism has an unbounded approximation ratio, we prove that this is not the case if the agents' type is sampled from a probability distribution.
This framework is also known as Bayesian Mechanism Design \cite{chawla2014bayesian,hartline2013bayesian}.
Our main contribution shows that it is possible to select a percentile mechanism that asymptotically behaves optimally, i.e. it minimizes the expected social objective.

\subsection{Our Contribution}
In this paper, we conduct a comprehensive investigation of the $k$-Facility Location Problem ($k$-FLP) from a Bayesian Mechanism Design perspective, where we assume that agents' positions on the line follow a distribution $\mu$ \cite{chawla2014bayesian,hartline2013bayesian}. 
We focus specifically on the class of percentile mechanisms \cite{sui2013analysis} and explore the conditions under which the Bayesian approximation ratio of these mechanisms -- defined as the ratio between the expected cost induced by a mechanism and the expected optimal cost -- is bounded.
We establish that each percentile mechanism exhibits different performances depending on the measure $\mu$, and we identify the optimal percentile mechanism tailored to a distribution $\mu$.
Our study establishes a connection between the $k$-FLP and a projection problem in the Wasserstein space. 
Through this connection, we import tools and techniques from Optimal Transport theory to approach the $k$-FLP.
In particular, we demonstrate that when the number of agents on the line tends to infinity, the ratio between the expected cost induced by the mechanism and the expected optimal cost converges to a bounded value.
Moreover, we characterize both the limit value of the ratio and the speed of convergence.
To retrieve these results, we make massive use of Bahadur's representation formula, which relates the $j$-th ordered statistic of a random variable to a suitable quantile of the probability distribution associated with the random variable
Finally, leveraging the characterization of the limit and its convergence rate, we derive a bound on the performances of percentile mechanisms for any finite number of agents.
We then tackle the problem of retrieving the best percentile mechanism tailored to a distribution $\mu$ and the number of facilities $k$.
We show that there always exists a percentile vector, namely $\vec v_\mu\in(0,1)^k$, that induces the optimal percentile mechanism, \textit{i.e.} a mechanism whose expected social cost is asymptotic to the optimal expected cost when the number of agents increases.
We characterize this vector as the solution to a system of $k$ equations and employ it to compute the optimal percentile vector associated with common probability measures, such as the Uniform and Gaussian distributions.
Lastly, we show that the optimal percentile vector is invariant under positive affine transformations of the probability measures describing agents.
In particular, $\vec v_\mu$ does not depend on the specific mean and variance of the distribution $\mu$.
To conclude the paper, we present a study on the stability of the optimal percentile vector. 
Specifically, let $\tilde{\mu}$ be an approximation of the true agents' distribution $\mu$. 
Additionally, let $\vec{v}_{\tilde{\mu}}$ and $\vec{v}_{\mu}$ represent the optimal percentile vectors associated with $\tilde{\mu}$ and $\mu$, respectively.
We demonstrate that when the agents are distributed according to $\mu$, the Bayesian approximation ratio limit of the percentile mechanism induced by $\vec{v}_{\tilde{\mu}}$ deviates from $1$ by an amount proportional to the infinity Wasserstein distance between $\mu$ and $\tilde{\mu}$.
The more precise the approximation of $\mu$, the better the asymptotic performance of the optimal percentile mechanism induced by $\vec v_{\mu}$ when the agents are distributed according to $\mu$.
% 
% {\color{red} Due to space limits, some proofs are deferred to Appendix.}

% 
% All the results discussed hold for the most common social objectives, including the Social Cost, the Maximum Cost, and the $l_p$ costs.
% % 
% To maintain focus in the main discussion, we defer the majority of the results concerning the Maximum and $l_p$ costs to Appendix \ref{app:B}.
% 
% Due to space limits, some proofs are deferred to Appendix \ref{app:A} for the sake of clarity and brevity.
% 

\subsection{Related Work}

The study of $k$-FLP research from an algorithmic mechanism design viewpoint was initiated by Procaccia and Tennenholtz in \cite{procaccia2013approximate}.
When $k=1,2$ there are several truthful mechanisms, such as the median mechanism \cite{black1948rationale} and its generalizations \cite{barbera2001strategy}, that achieve small constant worst-case approximation ratios.
When $k> 2$, however, these efficiency guarantees are much more negative: there are no truthful, deterministic and anonymous mechanisms with bounded approximation ratio \cite{fotakis2014power}.
It is worthy of notice however, that this impossibility result does not apply to randomized mechanisms \cite{fotakis2013strategyproof}, to instances where the number of agents is precisely equal to the number of facilities plus one, as shown in \cite{escoffier2011strategy}, and to problems in which facilities have capacity limits \cite{auricchio2024facility,auricchiocapacitated}.
The \textit{Percentile Mechanisms} are a class of mechanisms for the $k$-FLP that, similarly to the median mechanism, places the facilities at the locations of $k$ agents depending on their order on the line, \cite{sui2013analysis}.
Due to the dictatorial-like nature of these mechanisms, it is easy to build an \textit{ad hoc} instance of the $k$-FLP on which the optimal social cost is as small as we like, but the cost attained by the mechanisms is greater than a positive constant.
These instances depend on the percentile mechanism and carry little practical sense in applied contexts. 
Bayesian Mechanism Design is an alternative paradigm to evaluate the performances of a mechanism in which every agent's type is drawn from a known probability distribution, \cite{hartline2013bayesian,chawla2014bayesian}.
Consequentially, this defines a distribution over the set of inputs over which the mechanism is defined, allowing us to introduce the notion of expected cost of a mechanism.
To the best of our knowledge, the only other two papers studying the $k$-FLP in a Bayesian Mechanism Design framework, are \cite{auricchio2023extended}, where the $k$-Capacitated Facility Location Problem is considered, and \cite{zampetakis2023bayesian}, in which the authors study how to use the Lugosi-Mendelson median \cite{10.1214/17-AOS1639} to define approximately truthful mechanisms for the $1$-FLP. 
Other fields in which Bayesian Mechanism Design framework has been applied are: routing games \cite{gairing2005selfish}, combinatorial mechanisms based on $\epsilon$-greedy mechanisms \cite{lucier2010price}, and auction mechanism design problems \cite{chawla2007algorithmic,hartline2009simple}.

Over the past few decades, Optimal Transport (OT) methods have gradually found their application within the broad landscape of Theoretical Computer Science.
Notable examples include Computer Vision \cite{Rubner2000,auricchio2020equivalence,Pele2009}, Computational Statistics \cite{Levina2001}, Clustering \cite{auricchio2019computing}, and Machine Learning in general \cite{scagliotti2023normalizing,frogner2015learning,scagliotti23}.
% 
% Notable examples include Computer Vision \cite{Rubner2000,auricchio2020equivalence,Pele2009}, Computational Statistics \cite{Levina2001}, Clustering \cite{auricchio2019computing}, and Machine Learning in general \cite{scagliotti2023normalizing,frogner2015learning,scagliotti23}.
%
However, there has been limited advancement in applying OT theory to the field of mechanism design.
To the best of our knowledge, the only field related to mechanism design that has been explored using OT theory is auction design \cite{daskalakis2013mechanism}.
In their work, the authors demonstrated that the optimal auction mechanism for independently distributed items can be characterized by the Dual Formulation of an OT problem.
Moreover, they utilized this relationship to derive the optimal mechanism for various item classes, thereby establishing a fruitful application of OT theory in the context of mechanism design.

\section{Preliminaries}
\label{sec:preliminaries}

In this section, we fix the necessary notations on the $k$-Facility Location Problem ($k$-FLP), Bayesian Mechanism Design, and Optimal Transport (OT).
Furthermore, we recall the definition of the percentile mechanisms. 

\paragraph{The \texorpdfstring{$k$}{k}-Facility Location Problem.}

Given a set of self-interested agents $\mathcal{N}=[n]:=\{1,2,\dots,n\}$, we denote with $\mathcal{X}:=\{x_i\}_{i\in [n]}$ the set of their positions over $\erre$. Without loss of generality, assume that the agents are indexed such that the positions $x_i$'s are in non-decreasing order. 
We denote with $\vec x :=(x_1,x_2,\dots,x_n)\in \erre^n$ the vector containing the elements of $\mathcal{X}$. 
In this setting, if the $k$ facilities are located at the entries of the vector $\vec y :=(y_1,y_2,\dots,y_k)\in \erre^k$, an agent positioned in $x_i$ incurs a cost of $c_i( x_i,\vec y)=\min_{j\in [k]}|x_i-y_j|$ to access a facility.
In what follows, we will use $\vec y=(y_1,\dots,y_k)$ and the set of points $\{y_j\}_{j\in [k]}$, interchangeably. 
% 
% In particular, we say that the facilities are located at $\vec y$ instead of saying that the facilities are located at the entries of $\vec y$.}
% 
% Finally, a cost function is a map $C:\erre^n\times \erre^k\to [0,+\infty)$ that associates to every couple $(\vec x , \vec y)$ the overall cost of placing the facilities at $\vec y$ when the agents are located at the entries of $\vec x$.
% 
Given a vector $\vec x\in\erre^n$ containing the agents' positions, the \textit{Social Cost} ($\SCo$) of $\vec y$ is the sum of all the agents' utilities, that is $\SCo(\vec x, \vec y)=\sum_{i\in [n]}c_i(x_i,\vec y )$.
The $k$-\textit{Facility Location Problem}, consists in finding the locations for $k$ facilities that minimize the function $\vec y \to \SCo(\vec x,\vec y)$.
Given that multiplying the cost function by a constant does not alter the approximation ratio of the mechanisms, we rescaled the Social Cost as $\SCo(\vec x, \vec y)=\frac{1}{n}\sum_{i\in [n]}c_i(x_i,\vec y )$.

\paragraph{Mechanism Design and the Worst-Case analysis.}
A \textit{$k$-facility location mechanism} is a function $f:\erre^n\to \erre^k$ that takes the agents' reports $\vec x$ in input and returns a set of $k$ locations $\vec y$ for the facilities. 
In general, an agent may misreport its position if it this results in a set of facility locations such that the agent's incurred cost is smaller than reporting truthfully.
A mechanism $f$ is said to be \textit{truthful} (or \textit{strategyproof}) if, for every agent, its cost is minimized when it reports its true position. That is, $c_i(x_i,f(\vec x))\le c_i(x_i,f(\vec x_{-i},x_i'))$ for any misreport $x_i'\in \erre$, where $\vec x_{-i}$ is the vector $\vec x$ without its $i$-th component.
Albeit deploying a truthful mechanism instead of computing the optimal location prevents agents from misreporting their positions, it comes with a loss in terms of efficiency.
To evaluate this efficiency loss, Nisan and Ronen introduced the notion of approximation ratio of a truthful mechanism \cite{nisan1999algorithmic}.
Given a truthful mechanism $f$, its approximation ratio is defined as 
\begin{equation}
    \label{eq:worstcaseapproximationratio}
    ar(f):=\sup_{\vec x \in \mathbb{R}^{n}}\frac{SC_{f}(\vec x)}{SC_{opt}(\vec x)},
\end{equation}
where $SC_f(\vec x)$ is the Social Cost of placing the facilities at $f(\vec x)$ and $SC_{opt}(\vec x)$ is the optimal Social Cost achievable when the agents' report is $\vec x$.
In what follows, we will refer to the worst-case approximation ratio defined in \eqref{eq:worstcaseapproximationratio} as the approximation ratio.
Evaluating a mechanism $f$ from its approximation ratio is also known as the worst-case analysis of $f$.

\paragraph{Bayesian Analysis.}
In Bayesian Mechanism Design, we assume that the agents' types follow a probability distribution and study the performance of mechanisms from a probabilistic viewpoint.
Every agent's type is then described by a random variable $X_i$.
In what follows, we assume that every $X_i$ is identically distributed according to a law $\mu$ and independent from the other random variables.
A mechanism is said to be truthful if, for every agent $i$, it holds
\begin{equation}
    \EE_{\vec X_{-i}}[c_i(x_i,f(x_i,\vec X_{-i}))]\le \EE_{\vec X_{-i}}[c_i(x_i,f(x_i',\vec X_{-i}))] \quad\quad\quad \forall x_i\in\erre,
\end{equation}
where $x_i$ agent $i$'s true type, $\vec X_{-i}$ is the $(n-1)$-dimensional random vector that describes the other agents' type, and $\EE_{\vec X_{-i}}$ is the expectation with respect to the joint distribution of $\vec X_{-i}$.
Given $\beta\in\erre$, a mechanism $f$ is a $\beta$-approximation if $\EE[SC_f(\vec X_n)]\le \beta\,\EE[SC_{opt}(\vec X_n)]$ holds, so that the lower $\beta$ is, the better the mechanism is.
To unify the notation, we define the Bayesian approximation ratio for the Social Cost as the ratio between the expected Social Cost of a mechanism and the expected Social Cost of the optimal solution.
More formally, given a mechanism $f$, its Bayesian approximation ratio is defined as follows 
\begin{equation}
\label{eq:B_ar}    
    B_{ar}^{(n)}(f):=\frac{\EE[SC_f(\vec X_n)]}{\EE[SC_{opt}(\vec X_n)]},
\end{equation}
where the expected value is taken over the joint distribution of the vector $\vec X_n := (X_1,\dots,X_n)$.
Notice that, if $B_{ar}^{(n)}(f)<+\infty$, then $f$ is a $B_{ar}^{(n)}(f)$-approximation.
Since we consider only truthful mechanisms, in what follows we use $\vec x$ to denote the vector containing the agents' reports and the agents' real position interchangeably.
Moreover, we use the capital letter $\vec X_n$ to denote the random vector describing the agents' types.

\paragraph{The Percentile Mechanisms.}

The class of \textit{percentile mechanisms} has been introduced in \cite{sui2013analysis}. 
Given a vector $\vec v=(v_1,v_2,\dots,v_k)$, such that $0\le v_1\le v_2\le \dots\le v_k\le 1$, the percentile mechanism induced by $\vec v$, namely $\PMp$, proceeds as follows:
\begin{enumerate*}[label=(\roman*)]
    \item The mechanism designer collects all the reports of the agents, namely $\{x_1,\dots,x_n\}$ and reorders them  non-decreasingly.
    Without loss of generality, let us assume that the reports are already ordered in non-decreasing order, i.e. $x_i\le x_{i+1}$.
    \item The designer places the $k$ facilities at the positions $y_j=x_{i_j}$, where $i_j=\floor{(n-1)v_j}+1$.
\end{enumerate*}
Notice that, if the values $x_i$ are sampled from a distribution, the output of any percentile mechanism is composed by the $(\floor{(n-1)v_j}+1)$-th order statistics of the sample.
Percentile mechanisms are truthful whenever the cost of an agent placed at $x_i$ is $c_i=\min_{j\in [k]}|x_i-y_j|$, where $y_j$ are the position of the facilities. 
Thus, when $k>2$, the approximation ratio of any percentile mechanism becomes unbounded since the percentile mechanisms are also anonymous and deterministic, that is $ar(\PMp) = +\infty$ for every percentile vector $\vec v$.
Moreover, it is worth noting that since percentile mechanisms are truthful in the classic setting, they also retain their truthfulness within the Bayesian framework \cite{hartline2013bayesian}.
%
% Therefore, throughout this paper, there is no need to distinguish between the agents' types and their reported information.
% 

\paragraph{Basic Notions on Optimal Transport.}
In the following, we denote with $\PP(\erre)$ the set of probability measures over $\erre$. 
Given a measure $\gamma\in\PP(\erre)$, we denote with $spt(\gamma)\subset \erre$ the support of $\gamma$, that is, the smallest closed set $C\subset \erre$ such that $\gamma(C)=1$.
Furthermore, we denote with $\PP_k(\erre)$ the set of probability measures over $\erre$ whose support consists of $k$ points. 
That is, $\nu\in\PP_k(\erre)$ if and only if $\nu=\sum_{j=1}^k \nu_j\delta_{x_j}$, where $x_j\in \erre$ for every $j\in [k]$, $\nu_j\ge 0$ are real values such that $\sum_{j=1}^k\nu_j=1$, and $\delta_{x_j}$ is the Dirac's delta centered in $x_j$.
Given two measures $\alpha,\beta\in\PP(\erre)$, the Wasserstein distance between $\alpha$ and $\beta$ is defined as 
\begin{equation}
    \label{eq:wass_p}
    W_1(\alpha,\beta)=\min_{\pi\in\Pi(\alpha,\beta)}\int_{\erre\times \erre}|x-y|d\pi,
\end{equation}
where $\Pi(\alpha,\beta)$ is the set of probability measures over $\erre\times \erre$ whose first marginal is equal to $\alpha$ and the second marginal is equal to $\beta$, \cite{kantorovich2006translocation}.
Lastly, the infinity Wasserstein distance is defined  as $W_\infty(\alpha,\beta)=\min_{\pi\in\Pi(\alpha,\beta)}\max_{(x,y)\in spt(\pi)} |x-y|$.
It is well-known that both $W_1$ and $W_\infty$ are metrics over $\PP(\erre)$.
For a complete introduction to the Optimal Transport theory, we refer the reader to \cite{villani2009optimal}.

\paragraph{Basic Assumptions}

In the remainder of the paper, we tacitly assume that the underlying distribution $\mu$ satisfies the following properties:
\begin{enumerate*}[label=(\roman*)]
\item The measure $\mu$ is absolutely continuous. We denote with $\rho_\mu$ its density.
\item The support of $\mu$ is an interval, which can be bounded or not, and that $\rho_\mu$ is strictly positive on the interior of the support.
\item The density function  $\rho_\mu$ is differentiable on the support of $\mu$.
\end{enumerate*}
Notice that the cumulative distribution function (c.d.f.) $F_\mu$ of a probability measure $\mu$ satisfying these properties is locally bijective.
Thus the pseudo-inverse function of $F_\mu$, namely $F^{[-1]}_\mu$, is well-defined on $(0,1)$.

\section{The Bayesian Analysis of the Percentile Mechanism}
\label{Wass_formalism}

In this section, we study the percentile mechanisms in the Bayesian Mechanism Design framework.
Specifically, we consider a scenario where the agents' reports are drawn from a shared distribution $\mu$, which satisfies the basic assumptions outlined in Section \ref{sec:preliminaries}.
First, we establish a connection between the $k$-FLP and the Wasserstein distance and use it to investigate the convergence behaviour of the Bayesian approximation ratio as the number of agents tends to infinity.
%
% Thirdly, we demonstrate that, for any given probability distribution $\mu$, there exists a percentile vector $\vec v_\mu$, that yields an asymptotically optimal expected cost.
% % 
% Furthermore, we characterize $\vec v_\mu$ as the solution of a system of $k$ equations and study the properties of optimal percentile mechanism $\mathcal{PM}_{\vec v_\mu}$.
% %
% Finally, we demonstrate that computing the optimal percentile vector based on an approximation of $\mu$ results in an asymptotically quasi-optimal percentile mechanism, and we quantify the degree of suboptimality of these mechanisms in terms of the $W_\infty$ distance between $\mu$ and its approximation.
% 

\subsection{The \texorpdfstring{$k$}{k}-FLP as a Wasserstein Projection problem}

Given a vector $\vec x:=(x_1,x_2,\dots,x_n)$ containing the reports of $n$ agents, we define the measure $\mu_{\vec x}:=\frac{1}{n}\sum_{i=1}^n\delta_{x_i}$.
Using the map $\vec x \to \mu_{\vec x}$, we are able to associate any agents' profile to a probability measure in $\PP_n(\erre)\subset\PP(\erre)$.
Let us now consider the following minimization problem
\begin{equation}
    \label{eq:projection_problem_gen}
    \min_{\lambda \in \PP_k(\erre)}W_1(\mu_{\vec x},\lambda).
\end{equation}
% where $p\in [1,+\infty]$.
% 
Due to the metric properties of $W_1$, problem \eqref{eq:projection_problem_gen} is also known as the Wasserstein projection problem on $\PP_k(\erre)$.
Since $\PP_k(\erre)$ is closed with respect to any $W_1$ metric, any Wasserstein projection problem admits at least a solution \cite{ambrosio2005gradient}.
When $\mu_{\vec x}$ is clear from the context, we denote with $\nu^{(k,n)}$ the solution to problem \eqref{eq:projection_problem_gen}.
In general, given a measure $\zeta$, we say that $\nu$ is the projection of $\zeta$ over $\mathcal{S}\subset\PP(\erre)$ with respect to $W_1$ if $\nu\in\mathcal{S}$ and $W_1(\zeta,\nu)\le W_1(\zeta,\rho)$ for every $\rho\in\mathcal{S}$.
In particular, $\nu^{(k,n)}$ is the projection of $\mu_{\vec x}$ over $\PP_k(\erre)$ with respect to $W_1$.
The starting point of our Bayesian analysis of the percentile mechanisms connects the $k$-FLP to a Wasserstein projection problem.
In particular, the objective value of problem \eqref{eq:projection_problem_gen} is the same as the objective value of the $k$-FLP.

\begin{theorem}
\label{thm:SocialCost_projection}
Let $\vec x$ be the reports of $n$ agents.
Let $\vec y$ be the solution to the $k$-FLP, i.e. the facility locations that minimize the Social Cost.
Then the set $\{y_j\}_{j\in [k]}$ is the support of a measure $\nu^{(k,n)}$ that solves problem \eqref{eq:projection_problem_gen}.
Moreover, we have that
\[
\SCo_{opt}(\vec x) = W_1(\mu_{\vec x},\nu^{(k,n)}) = \min_{\lambda \in\PP_k(\erre)}W_1(\mu_{\vec x},\lambda).
\]
Vice-versa, if $\nu\in\PP_k(\erre)$ is a solution to problem \eqref{eq:projection_problem_gen}, then its support $\{y_j\}_{j\in [k]}$ is a solution to the $k$-FLP.
% 
% Similarly, if $p=\infty$, the support of the solution to problem \eqref{eq:projection_problem_gen} is the optimal location with respect to the Maximum Cost.
% % 
% Vice-versa, the entries of every optimal location vector $\vec y$ for the $k$-FLP problem with respect to the Maximum Cost is the support of a minimizer of problem \eqref{eq:projection_problem_gen} for $p=\infty$.
% % 
% Finally, the same implications hold for every $p\in (1,\infty)$.
% % 
% In this case, the solution is taken with respect to the $l_p$ Cost.
\end{theorem}

\begin{proof} 
Let $\vec x$ be the vector containing the reports of $n$ agents, and let $ \vec y$ be the vector containing the optimal location for $k$ facilities when the agents are located according to $\vec x$.
Without loss of generality, we assume that the closest facility to each agent $x_i$ is unique so that the sets $A_j$, defined as $A_j :=\Big\{ x_i \; : \; \min_{l\in [k]}|x_i-y_l|=|x_i-y_j| \Big\}$,
are well-defined and disjoint.
First, we show that, given an optimal facility location $\vec y$, it is possible to retrieve a measure $\nu\in\PP_k(\erre)$ that solves the projection problem \eqref{eq:projection_problem_gen} and whose support is $\{y_j\}_{j\in [k]}$.
For every $y_j$, let us set $\nu_j=\frac{\ell_j}{n}$, where $\ell_j:=|A_j|$ is the number of agents whose closest facility is located at $y_j$.
We then set $\nu=\sum_{j\in [k]}\nu_j\delta_{y_j}$.
Since $A_j$ are disjoint sets, we have $\nu\in\PP_k(\erre)$.
Let us now consider the transportation plan, namely $\pi$, between $\mu_{\vec x}$ and $\nu$ defined as
\[
\pi_{i,j}:=\pi_{x_i,y_j}=\begin{cases}
    \frac{1}{n}\quad\quad \text{if}\quad x_i\in A_j\\
    0 \quad \quad \text{otherwise.}
\end{cases}
\]
Since according to $\pi$ every agent goes to its closest facility, $\pi$ is optimal, thus we have $W_1(\mu_{\vec x},\nu)=\sum_{i\in [n],j\in [k]}|x_i-y_j|\pi_{i,j}=\frac{1}{n}\sum_{j \in [k]}\sum_{x_i\in A_j}|x_i-y_j|$.
We now show that $\nu$ solves problem \eqref{eq:projection_problem_gen}.
Toward a contradiction, let us assume that $\tilde \nu =\sum_{j=1}^{k}\tilde \nu_j\delta_{\tilde y_j}  \in \PP_k(\erre)$ is such that $W_1(\mu_{\vec x},\tilde \nu)<W_1(\mu_{\vec x},\nu)$.
Let us define the partition of agents $A_j'$ related to the set of points $\{y_j'\}_{j\in [k]}$.\footnote{Again, without loss of generality, we can assume that the facility that is closest to a given agent is unique.}
Then we have
\begin{equation}
    \frac{1}{n}\sum_{j \in [k]}\sum_{x_i\in A_j'}|x_i-y_j'|=W_1(\mu_{\vec x},\tilde \nu_j)<W_1(\mu_{\vec x},\nu)=\frac{1}{n}\sum_{j \in [k]}\sum_{x_i\in A_j}|x_i-y_j|,
\end{equation}
which contradicts the optimality of $\vec y$, proving the first part of the Theorem.
For the inverse implication, it suffices to repeat the same argument backwards.
Indeed, let $\nu'$ be a solution to the $W_1$ Projection problem.
Toward a contradiction, let us assume that the support of $\nu'$ is not a solution to the $k$-FLP.
Then, given a solution to the $k$-FLP problem, we can use the argument used in the first part of the proof  to build a new measure that has a lower cost than $\nu'$, which would contradict the optimality of the initial solution.
\qed
\end{proof}

By restricting the set on which the projection problem is defined, we retrieve a similar characterization for the cost of any $k$-facility location mechanism.

\begin{theorem}
\label{thm:swich_bound}
Let $f:\erre^n\to \erre^k$ be a $k$-facility location mechanism.
Then, the following identity holds
% , the Social Cost of $f$ is the minimal distance between $\mu_{\vec x}$ and a closed subset of $\PP_k(\erre)$ with respect to $W_1$, that is
\begin{equation}
    \label{eq:percentile_SC}
\SCo_f(\vec x) = \min_{\{\lambda_j\}_{j\in [k]}\subset \erre} W_1(\mu_{\vec x},\lambda),
\end{equation}
where $\lambda=\sum_{j\in [k]}\lambda_j\delta_{y_{j}}$ and $\vec y =(y_1,y_2,\dots,y_k)=f(\vec x)$.
% 
% Similarly, the $l_p$ Cost or the Maximum Cost of the mechanism $f$ can be expressed as a projection problem described in \eqref{eq:percentile_SC} by replacing $W_1$ with $W_1$ or $W_\infty$, respectively.
\end{theorem}

\begin{proof}
Let $f$ be a mechanism, $\vec x$ the vector containing the reports of $n$ agents, and let us denote with $\vec y$ the vector containing the positions returned by the mechanism $f$, so that $\vec y=f(\vec x)$. 
For every $j\in [k]$, let us denote with $A_j$ the set of agents that are closer to the facility placed at $y_j$.
Without loss of generality, we assume that every $A_j$ is disjoint from the other so that $A_j\cap A_r=\emptyset$ for every $j\neq r$.
Let us now define $\nu^{(n)}$ as
\[
\nu^{(n)}=\sum_{j\in [k]}\nu_j^{(n)}\delta_{y_j}
\]
where $\nu_j^{(n)}=\frac{\ell_j}{n}$ and $\ell_j=|A_j|$.
We now show that $\nu^{(n)}$ is a solution to problem \eqref{eq:percentile_SC}.
Indeed, the discrete probability measure $\pi$ is defined as
\[
\pi_{i,j}:=\pi_{x_i,y_j}=\begin{cases}
\frac{1}{n} \quad\quad \text{if}\quad x_i\in A_j\\
0 \quad \quad \text{otherwise,}
\end{cases}
\]
is a transportation plan between $\mu_{\vec x}$ and $\nu^{(n)}$.
Furthermore, since according to $\pi$ every agent goes to its closest facility, we have
\[
\sum_{i\in [n]}|x_i-y_j|\pi_{i,j}=W_1(\mu_{\vec x},\nu^{(n)}).
\]
Finally, if $\Tilde{\nu}$ is such that $spt(\Tilde{\nu})=spt(\nu)=\{y_j\}_{j\in [k]}$ and $W_1(\mu_{\vec x},\Tilde{\nu})<W_1(\mu_{\vec x},\nu^{(n)})$, we infer that there exists at least one agent that can be reallocated to a closer facility, which would contradict the definition of $A_j$.
\qed
\end{proof}

Notice that the projection problem \eqref{eq:percentile_SC} is a further restricted version of the projection problem \eqref{eq:projection_problem_gen}.
Indeed, in \eqref{eq:projection_problem_gen}, the support of the solution can be any subset of $\erre$ containing $k$ elements, while in \eqref{eq:percentile_SC}, the support of the solution is fixed by the mechanism $f$.

\subsection{The Bayesian Analysis of the Percentile Mechanisms}
\label{sec:bayesian}

In this section, we use the results presented in Theorem \ref{thm:SocialCost_projection} and Theorem \ref{thm:swich_bound} to study the limiting behaviour of the Bayesian approximation ratio of any percentile mechanism $\PMp$ with $\vec{v} \in (0,1)^k$.
From Theorem \ref{thm:SocialCost_projection}, the $k$-FLP is equivalent to a projection problem in the space of probability distributions with respect to $W_1$.
It is well-known that, in order to ensure that the $W_1$ distance between two measures is finite, both the measures must have a finite first moment \cite{villani2009optimal}.
We recall that a measure $\mu$ has a finite first moment if
\begin{equation}
\label{eq:finite_moment}
    \int_{\mathbb{R}}|x|d\mu < +\infty.
\end{equation}

\begin{lemma}
\label{lmm:lemma1}
Let $\vec X_n:=(X_1,X_2,\dots,X_n)$ be the random vector describing the reports of $n$ i.i.d. agents distributed as $\mu$.
If $\mu$ satisfies \eqref{eq:finite_moment}, then, for every $k\in \mathbb{N}$, we have that $\EE[\SCo_{opt}(\vec X_n)]$ converges to $W_1(\mu,\nu^{(k)})$ as $n\to\infty$, where $\nu^{(k)}$ is the solution to the following projection problem
\begin{equation}
    \label{eq:projection_real}
    \min_{\lambda\in \PP_k(\erre)}W_1(\mu,\lambda).
\end{equation}
In particular, we have that $\EE[\SCo_{opt}(\vec X_n)]$ is strictly positive for $n$ large enough.
\end{lemma}

\begin{proof}
Let $\nu^{(k,n)}$ be the solution to problem \eqref{eq:projection_problem_gen} and let $\nu^{(k)}$ be the solution to problem \eqref{eq:projection_real}.
Owing to the triangular inequality and to the properties of the projection problem, we have
\[
W_1(\mu_{\vec x},\nu^{(k,n)})\le W_1(\mu_{\vec x},\nu^{(k)})\le W_1(\mu_{\vec x},\mu)+W_1(\mu,\nu^{(k)})
\]
and, similarly
\[
W_1(\mu,\nu^{(k)})\le W_1(\mu,\nu^{(k,n)})\le W_1(\mu,\mu_{\vec x})+W_1(\mu_{\vec x},\nu^{(k,n)}),
\]
which implies $|W_1(\mu,\nu^{(k)})-W_1(\mu_{\vec x},\nu^{(k,n)})|\le W_1(\mu,\mu_{\vec x})$ and thus 
\[
\EE[|W_1(\mu,\nu^{(k)})-W_1(\mu_{\vec x},\nu^{(k,n)})|]\le \EE[W_1(\mu,\mu_{\vec x})].
\]
Since $\lim_{n\to\infty}\EE[W_1(\mu,\mu_{\vec x})]=0$, see \cite{bobkov2019one}, we infer that $\EE[W_1(\mu_{\vec x},\nu^{(k,n)})]$ converges to $W_1(\mu,\nu^{(k)})$ as $n\to \infty$.
Finally, since $W_1(\mu,\nu^{(k)})$ is strictly positive, for $n$ large enough, we have $\EE[W_1(\mu_{\vec x},\nu^{(k,n)})]$ is strictly positive as well.

\qed\end{proof}

It is worthy of notice that in the proof of Lemma \ref{lmm:lemma1}, we have shown a slightly stronger result: the random variable $\SCo_{opt}(\vec X_n)$ converges with respect to the $L^1$ norm to the constant value $W_1(\mu,\nu^{(k)})$.
Moving on to the limit cost of the mechanism, we observe that the set characterizing the projection problem \eqref{eq:percentile_SC} is dependent on the output of the percentile mechanism.
Hence, the argument used to prove Lemma \ref{lmm:lemma1} cannot be directly applied in this case.
However, by employing a more sophisticated construction and leveraging the convergence properties of the $k$-th order statistics, it is possible to identify the limit of $\EE[\SCo_{\vec v}(\vec X_n)]$ and ensure convergence by imposing mild assumptions on the percentile vector $\vec v$.

\begin{lemma}
\label{lmm:lemma2}
Let $\mu$ be a measure that satisfies \eqref{eq:finite_moment}.
Given $k\in \mathbb{N}$, let $\vec v\in (0,1)^k$ be a percentile vector.
Then, $\EE[SC_{\vec v}(\vec X_n)]$ converges to $W_1(\mu,\nu_{Q_{\vec v}})$, where $\nu_{Q_{\vec v}}$ is defined as
\begin{equation}
    \label{eq:nu_qp}
    \nu_{Q_{\vec v}}:=\sum_{i=1}^k (F_\mu(z_i)-F_\mu(z_{i-1}))\delta_{F_{\mu}^{[-1]}(v_i)},
\end{equation}
where $z_i=\frac{(F_{\mu}^{[-1]}(v_i)+F_{\mu}^{[-1]}(v_{i+1}))}{2}$ for $i=1,\dots,k-1$, $z_0=\inf_{x\in I}x$, and $z_k=\sup_{x\in I}x$, $F_\mu$ is the cumulative distribution function of $\mu$, and $F_\mu^{[-1]}$ is the pseudo-inverse function related to $\mu$.
\end{lemma}

\begin{proof}
First, we notice that the measure \eqref{eq:nu_qp} is well-defined since there exists a $j$ such that $v_j\neq 0,1$.
Let $\vec v=(v_1,\dots, v_k)$ be a percentile vector, $\vec x$ be the vector containing the reports of the agents, $\nu^{(k,n)}$ be the solution to problem \eqref{eq:percentile_SC}, and let $\vec y$ be the vector containing the facility positions returned by the percentile mechanisms, so that $\nu^{(k,n)}=\sum_{j\in [k]}(\nu^{(k,n)})_j\delta_{y_j}$.
% 
% Without loss of generality, we assume that $y_\ell\le y_{\ell+1}$ for every $\ell\in[k-1]$, so that $\nu^{(k,n)}=\sum_{j\in [k]}(\nu^{(k,n)})_j\delta_{y_j}$.
% 
To lighten-up the notation, we set $\nu_{Q_{\vec v}}:=\nu_{Q}$, thus $\nu_Q:=\sum_{j\in[k]}(\nu_Q)_j\delta_{F_\mu^{[-1]}(v_j)}$, where $(\nu_Q)_j := (F_\mu(z_j)-F_\mu(z_{j-1}))$, where $z_0=-\infty$, $z_k=+\infty$, and $z_i=\frac{y_i+y_{i+1}}{2}$ for every $i=2,\dots, k-1$.
We now show that $\nu_Q$ is the solution to the following minimization problem
\begin{equation}
\label{eq:minproblemnuQ}
    \min_{\{\lambda_j\}_{j\in[k]}}W_1(\lambda,\mu),
\end{equation}
where $\lambda=\sum_{j=1}^k\lambda_j\delta_{F^{[-1]}(v_j)}$.
We can rewrite the $W_1$ distance between $\mu$ and $\nu_{Q}$ as it follows
\[
W_1(\mu,\nu_{Q})=\sum_{j=1}^k\int_{F^{[-1]}(\sum_{i=1}^{k}(\nu_{Q})_i)}^{F^{[-1]}(\sum_{i=1}^{k+1}(\nu_{Q})_i)}|x-F^{[-1]}(v_j)|d\mu.
\]
By definition of $\nu_{Q}$ we have that $\sum_{i=1}^{j}(\nu_{Q_{\vec v}})_i=F_\mu(z_{j})$, thus
\begin{equation}
\label{eq:nuQclsoest}
    W_1(\mu,\nu_{Q})=\sum_{j=0}^{k}\int_{z_{j}}^{z_{j+1}}|x-F^{[-1]}(v_j)|d\mu=\int_{-\infty}^{+\infty}\min_{j\in[k]}|x-F^{[-1]}(v_j)|d\mu,
\end{equation}
where we used the fact that $z_0=-\infty$, $z_k=+\infty$, and $z_i=\frac{F^{[-1]}(v_i)+F^{[-1]}(v_{i+1})}{2}$ for every $i=2,\dots, k-1$.
Thus every point in the support of $\mu$ is assigned to its closest facility, thus $\nu_{Q}$ is a solution to \eqref{eq:minproblemnuQ}.
We are now ready to study the convergence of $\EE[SC_{\vec v}(\vec X_n)]$.
For every $n\in \mathbb{N}$, let us define $\gamma_n=\sum_{j\in [k]}(\nu_Q)_j\delta_{y_j}$, where $y_j$ is the $j$-th point in the support of $\nu^{(k,n)}$.
By a similar argument to the one used to prove Lemma \ref{lmm:lemma1}, we have
\begin{equation*}
    W_1(\mu_{\vec x},\nu^{(k,n)})\le W_1(\mu_{\vec x},\gamma_n)\le W_1(\mu_{\vec x},\mu)+W_1(\mu,\nu_Q)+W_1(\nu_Q,\gamma_n).
\end{equation*}
For every $n\in \mathbb{N}$, let us now define $\eta_n:=\sum_{j\in [k]}(\nu^{(k,n)})_j\delta_{F_\mu^{[-1]}(v_j)}$.
We then have
\[
W_1(\mu,\nu_Q)\le W_1(\mu,\eta_n)\le W_1(\mu,\mu_{\vec x})+W_1(\mu_{\vec x},\nu^{(k,n)})+W_1(\nu^{(k,n)},\eta_n).
\]
Since $W_1(\nu_Q,\gamma_n),W_1(\nu^{(k,n)},\eta_n)\ge 0$, we infer that
\begin{equation}
\label{eq:estimation_above_lemma_2_start}
\EE[|W_1(\mu,\nu_Q)-W_1(\mu_{\vec x},\nu^{(k,n)})|]\le \EE[W_1(\mu,\mu_{\vec x})]+\EE[W_1(\nu_Q,\gamma_n)]+\EE[W_1(\nu^{(k,n)},\eta_n)].
\end{equation}
From \cite{bobkov2019one}, we have that $\lim_{n\to\infty}\EE[W_1(\mu,\mu_{\vec x})]= 0$.
To conclude the thesis, we need to prove that both $\EE[W_1(\nu_Q,\gamma_n)]$ and $\EE[W_1(\nu^{(k,n)},\eta_n)]$ go to zero as $n\to \infty$.
If we express $\nu^{(k,n)}$ and $\eta_n$ explicitly, we infer that $\EE[W_1(\nu^{(k,n)},\eta_n)]$ converges to zero if the $(\floor{v_j(n-1)}+1)$-th quantile converges to $F_\mu^{[-1]}(v_j)$ with respect to the $l_1$ norm, which can be done using Bahadur's representation formula \cite{de1979bahadur} and leveraging our hypothesis on the regularity of $\mu$.
First, we prove that both $\EE[W_1(\nu_Q,\gamma_n)]$ and $\EE[W_1(\nu^{(k,n)},\eta_n)]$ go to zero as $n\to \infty$.
% 
% we conclude the proof, since it would imply that $\EE[W_1(\mu_{\vec x},\nu^{(k,n)})]$ converges to $W_1(\mu,\nu_Q)$ as $n\to \infty$.
% 
Without loss of generality, we show that $\EE[W_1(\nu_Q,\gamma_n)]\to 0$ as $n\to \infty$ \footnote{Indeed, the same argument can be adapted to show that also $\EE[W_1(\nu^{(k,n)},\eta_n)]$ goes to $0$}.
Given a set of reports $\vec x = (x_1,\dots,x_n)$, let $\vec y =(y_1,\dots,y_k)$ be the facility locations returned by the percentile mechanism $\PMp$.
Since every $y_j$ is the $(\floor{(n-1)v_j }+1)$-th statistic of the sample vector $\vec x$, we have that
\begin{align*}
    \EE[W_1(\nu_Q,\gamma_n)]&\le \sum_{j\in [k]}(\nu_Q)_j \EE[|X_{\floor{(n-1)v_j}+1}-F^{[-1]}_\mu(v_j)|]\\
    &\le \sum_{j\in [k]}\EE[|X_{\floor{(n-1)v_j}+1}-F^{[-1]}_\mu(v_j)|],
\end{align*}
where $X_{\floor{(n-1)v_j}+1}$ is the $(\floor{(n-1)v_j}+1)$-th order statistic of the sample vector $\vec x$.
Using Bahadur's representation formula (see \cite{10.1214/aoms/1177699450} and \cite{de1979bahadur}), we have that
\begin{equation}
    \label{eq:bahadur}
    X_{\floor{(n-1)v_j}+1}-F_\mu^{[-1]}(v_j)=\frac{S_n(F_\mu^{[-1]}(v_j))-v_j)}{\rho_\mu(F_\mu^{[-1]}(v_j))}+R_n,
\end{equation}
where $R_n$ is the rest of the Bahadur's formula, for which holds $R_n\le O(n^{-\frac{3}{4}})$ with probability $1$, $\rho_\mu$ is the density of $\mu$, and $S_n(t)=\frac{1}{n}\sum_{i=1}^n \mathbb{I}(X_i)_{\{X_i\le t\}}$ where 
\[
\mathbb{I}(X_i)_{\{X_i\le t\}}=\begin{cases}
    1 \quad \text{if}\quad X_i\le t,\\
    0 \quad \text{otherwise.}    
\end{cases}
\]
By taking the absolute value and the expected value of both sides of \eqref{eq:bahadur}, we have
\begin{equation}
    \label{eq:lefthandsidebahdur}
    \EE[|X_{\floor{(n-1)v_j}+1}-F_\mu^{[-1]}(v_j)|]\le \EE\Bigg[\Bigg|\frac{S_n(F_\mu^{[-1]}(v_j))-v_j}{\rho_\mu(F_\mu^{[-1]}(v_j))}\Bigg|\Bigg]+\EE[|R_n|].
\end{equation}
To conclude, we need to prove that the right-hand side of equation \eqref{eq:lefthandsidebahdur} goes to zero.
First, we observe that $\frac{S_n(F_\mu^{[-1]}(v_j))-v_j)}{\rho_\mu(F_\mu^{[-1]}(v_j))}$ converges to $0$ almost surely (use the Law of Large Numbers), and it is bounded.
This, combined with the i.i.d. assumption and condition \eqref{eq:finite_moment}, implies that $\frac{S_n(F_\mu^{[-1]}(v_j))-v_j)}{\rho_\mu(F_\mu^{[-1]}(v_j))}$ converges to $0$ with respect to the $L^1$ norm, since
\begin{equation}
    \EE\Bigg[\Bigg|\frac{S_n(F_\mu^{[-1]}(v_j))-v_j}{\rho_\mu(F_\mu^{[-1]}(v_j))}\Bigg|\Bigg]\le K\EE\Big[\Big|\sum_{i=1}^n\frac{(Z_i-p)}{n}\Big|\Big],
\end{equation}
where $\{Z_i\}_{i\in [n]}$ is a group of i.i.d. Bernoulli random variables that are equal to $1$ with probability $p$ and equal to $0$ with probability $1-p$.
By Jensen's inequality, we get
\[
\EE\Big[\Big|\sum_{i=1}^n\frac{(Z_i-p)}{n}\Big|\Big]\le \frac{1}{n}\Big(\EE[\sum_{i=1}^n|Z_i-p|^2]\Big)^{\frac{1}{2}}\le \frac{1}{\sqrt{n}}Var(Z_1)
\]
and thus $\EE\Big[\Big|\frac{S_n(F_\mu^{[-1]}(v_j))-v_j}{\rho_\mu(F_\mu^{[-1]}(v_j))}\Big|\Big]\le O(n^{-\frac{1}{2}})$.
Through a more delicate argument, we are able to show that $\EE[|R_n|]\le O(n^{-\frac{3}{4}})$.
% 
% To keep the discussion on track, we report this part of the proof in Appendix \ref{app:A}.
% 
Notice that we have shown that $\EE[|\SCo_{\vec v}(\vec X_n)-W_1(\mu,\nu_Q)|]\le \EE[W_1(\mu_{\vec x},\mu)]+O(n^{-\frac{1}{2}})$.
Owing once again to the convergence results in \cite{bobkov2019one}, we conclude the proof.
We now prove that $R_n$ converges to $0$ with respect to the $L^1$ norm.
% , that is $\EE[W_1(\nu^{(k)},\gamma_n)]\to 0$.
% 
We already know that $R_n$ converges to $0$ almost surely.
If we show that $R_n$ is also uniformly integrable, \textit{i.e.}
\begin{equation}
    \label{eq:unif_int_rest}
    \lim_{a\to +\infty}\sup_n \int_{a}^{+\infty}|R_n|dx=0
\end{equation}
and that $\sup_n\EE[|R_n|]<+\infty$, we infer from the Vitali convergence theorem that $R_n$ converges with respect to the $L^1$ norm.
More precisely, we show that $R_n$ converges to zero with respect to the $L^1$ norm and $\EE[R_n]\le O(n^{-\frac{3}{4}})$, which proves that $\EE[W_1(\nu_Q,\gamma_n)]\le O(n^{-\frac{1}{2}})$, thus $\EE[W_1(\nu_Q,\gamma_n)]\to 0$.
% 

% % 
% We first show that $\EE[W_1(\nu^{(k,n)},\eta_n)]\to 0$ as $n\to \infty$.
% % 
% From a similar argument, we have that
% \begin{align*}
%     \EE[W_1(\nu^{(k,n)},\eta_n)]&\le \sum_{j\in [k]}\nu^{(k,n)}_j \EE[|X_{\floor{(n-1)v_j}+1}-F^{[-1]}_\mu(v_j)|]\\
%     &\le \sum_{j\in [k]}\EE[|X_{\floor{(n-1)v_j}+1}-F^{[-1]}_\mu(v_j)|],
% \end{align*}
% % \]
% which, by the same argument we used to show $\EE[W_1(\nu_Q,\gamma_n)]\le O(n^{-\frac{1}{2}})$ as $n\to \infty$, allows us to conclude $\lim_{n\to \infty}\EE[W_1(\nu^{(k,n)},\eta_n)]\to 0$.
% % 
% To conclude, we notice that $\EE[|\SCo_{\vec v}(\vec X_n)-W_1(\mu,\nu_Q)|]\le \EE[W_1(\mu_{\vec x},\mu)]+O(n^{-\frac{1}{2}})$.
% 
% We complement the proof of Lemma \ref{lmm:lemma2} by showing that $R_n$ is uniformly integrable, thus $R_n$ converges with respect to the $L^1$ norm.
% 
From Bahadur's formula \eqref{eq:bahadur}, we have that
\[
R_n=X_{\floor{(n-1)v_j}+1}-F_\mu^{[-1]}(v_j)-\frac{S_n(F_\mu^{[-1]}(v_j))-v_j)}{\rho_\mu(F_\mu^{[-1]}(v_j))},
\]
which proves that $\sup_n\EE[|R_n|]<+\infty$, since $\sup_{n}\EE[|X_{\floor{(n-1)v_j}+1}|]<+\infty$\footnote{This follows from the fact that sample quantiles are asymptotically unbiased} and the other quantities are bounded random variables.
We now prove \eqref{eq:unif_int_rest}.
Since $\frac{S_n(F_\mu^{[-1]}(v_j))-v_j)}{\rho_\mu(F_\mu^{[-1]}(v_j))}$ and $F_\mu^{[-1]}(v_j)$ are both bounded, it suffices to show that the sequence $X_{\floor{(n-1)v_j}+1}$ is uniformly integrable.
It is well-known that the density of the $(\floor{(n-1)v_j}+1)$-th order statistic of an absolutely continuous random variable is given by the formula
\begin{equation}
    \label{eq:density_orderstatistics}
    \frac{n!}{(n-\floor{(n-1)v_j})!(\floor{(n-1)v_j})!}\rho_\mu(x)[F_\mu(x)]^{\floor{(n-1)v_j}}[1-F_\mu(x)]^{n-\floor{(n-1)v_j}}
\end{equation}
To prove that the sequence $X_{\floor{(n-1)v_j}+1}$ is uniformly integrable we then need to prove that
\[
\lim_{a\to +\infty}\sup_{n}\int_{a}^{+\infty}\frac{n!\;x\;\rho_\mu(x)[F_\mu(x)]^{\floor{(n-1)v_j}}[1-F_\mu(x)]^{n-\floor{(n-1)v_j}}}{(n-\floor{(n-1)v_j})!(\floor{(n-1)v_j})!}dx=0.
\]
Since $X$ has finite expected value, namely $m$, and every c.d.f. is increasing, we get
\begin{align}
\label{eq:gsfgfsgs}
  \nonumber  \int_{a}^{+\infty}&\frac{n!\;x\;\rho_\mu(x)[F_\mu(x)]^{\floor{(n-1)v_j}}[1-F_\mu(x)]^{n-\floor{(n-1)v_j}}}{(n-\floor{(n-1)v_j}-1)!(\floor{(n-1)v_j})!}dx\\
    &\le  \frac{n!\;m}{(n-\floor{(n-1)v_j}-1)!(\floor{(n-1)v_j})!}[1-F_\mu(a)]^{n-\floor{(n-1)v_j}}.
\end{align}
From the properties of the binomial coefficient, we have that
\[
\frac{n!}{(n-\floor{(n-1)v_j}-1)!(\floor{(n-1)v_j})!}[1-F_\mu(a)]^{n-\floor{(n-1)v_j}}\le \alpha_n [1-F_\mu(a)]^{(1-q)n},
\]
where $q\in (0,1)$ is such that $\max_{j\in [k]}v_j<q$ and
\begin{equation}
    \label{eq:dafsdf}
    \alpha_n=\begin{cases}
    \frac{n!}{(\frac{n}{2})!(\frac{n}{2})!} \quad \text{if} \quad n \;\text{is even},\\
    \frac{n!}{(\frac{n-1}{2})!(\frac{n+1}{2})!} \quad \text{otherwise.}
\end{cases}
\end{equation}
To conclude the proof, we show that there exists a value $a'$ such that for every $a>a'$ the sum of $\alpha_n(1-F_\mu(a))^{(1-q)n}$ converges; hence $\alpha_n(1-F_\mu(a))^{(1-q)n}$ is infinitesimal for every $a>a'$.
Indeed, using the ratio test criteria, we get that, for odd $n$
\begin{align*}
    \lim_{n\to \infty} \frac{(n+1)!}{(\frac{n+1}{2})!(\frac{n+1}{2})!}&\frac{(\frac{(n+1)}{2})!(\frac{(n-1)}{2})!}{n!}\frac{[1-F_\mu(a)]^{(1-q)(n+1)}}{[1-F_\mu(a)]^{(1-q)n}}=2[1-F_\mu(a)]^{(1-q)},
\end{align*}
thus, there exists a large enough value $a$ for which
\[
2[1-F_\mu(a)]^{(1-q)} < 1.
\]
Similarly, we show that also for even $n$ we infer the same conclusion, proving the uniformly integrability of $R_n$.
In particular, $R_n$ converges to zero with respect to the $L^1$ norm.
Moreover, notice that $\EE[R_n]\le O(n^{-\frac{3}{4}})$, which proves that  $\EE[W_1(\nu_Q,\gamma_n)]\le O(n^{-\frac{1}{2}})$, thus $\EE[W_1(\nu_Q,\gamma_n)]\to 0$.
A similar argument allows us to handle $\EE[W_1(\nu_Q,\gamma_n)]$ and conclude the proof.
\qed
\end{proof}

By combining the convergence results shown in Lemma \ref{lmm:lemma1} and \ref{lmm:lemma2}, we infer that the Bayesian approximation ratio of any $\PMp$ converges to a bounded quantity.

\begin{theorem}
% (Limit of the Bayesian approximation ratio)
\label{thm:limit}
Let $\vec X_n$ be a random vector of $n$ i.i.d. variables distributed as $\mu$ and let $\vec v\in (0,1)^k$ be a percentile vector.
If $\mu$ satisfies \eqref{eq:finite_moment}, we have
\[
\lim_{n\to \infty}\frac{\EE[\SCo_{\vec v}(\vec X_n)]}{\EE[\SCo_{opt}(\vec X_n)]}=\frac{W_1(\mu,\nu_{Q_{\vec v}})}{W_1(\mu,\nu^{(k)})}.
\]
\end{theorem}

\begin{proof}
    From Lemma \ref{lmm:lemma1}, we have that $\EE[\SCo_{opt}(\vec X_n)]$ converges to $W_1(\mu,\nu^{(k)})$, where $\nu^{(k)}$ is the solution to problem \eqref{eq:projection_real}.
    Since $\mu$ is absolutely continuous and $\nu^{(k)}$ is a discrete measure, we have that $W_1(\mu,\nu^{(k)})>0$.
    From Lemma \ref{lmm:lemma2}, have that $\EE[\SCo_{\vec v}(\vec X_n)]$ converges to $W_1(\mu,\nu_{Q_{\vec v}})$, where $\nu_{Q_{\vec v}}$ is defined as in \eqref{eq:nu_qp}.
    Finally, since both quantities are well-defined and finite, we have
    \[
    \lim_{n\to \infty}\frac{\EE[\SCo_{\vec v}(\vec X_n)]}{\EE[\SCo_{opt}(\vec X_n)]}=\frac{W_1(\mu,\nu_{Q_{\vec v}})}{W_1(\mu,\nu^{(k)})}<+\infty,
    \]
    which concludes the proof.
\qed\end{proof}

Theorem \ref{thm:limit} ensures that the limit of the Bayesian approximation ratio of any percentile mechanisms is equal to a quantity that depends only on $\mu$, $k$, and $\vec v$.
For an illustration, we compute this quantity for a generic $k$ and $\vec v$, and $\mu$ is the uniform distribution over $[0,1]$.
\begin{example}
    \label{example1}
    Let $\vec v =(v_1,\dots,v_k)$ be a percentile vector and let the underlying distribution $\mu$ be the uniform distribution over $[0,1]$.
    The measure $\nu_{Q_{\vec v}}$ is then defined as $\nu_{Q_{\vec v}}:=\sum_{i=1}^k \frac{v_{i+1}-v_{i-1}}{2}\delta_{v_i}$, where $v_0=0$ and $v_{k+1}=1$.
    It is easy to see that the projection of $\mu$ over $\PP_k(\erre)$ is $\nu^{(k)}:=\frac{1}{k}\sum_{j=1}^k\delta_{\frac{2j-1}{2k}}$.
    From a simple computation, we infer that $W_1(\mu,\nu_{Q_{\vec v}})=\sum_{i=1}^k\Big[ \frac{(v_{i+1}-v_i)^2+(v_i-v_{i-1})^2}{2} \Big]$
    and $W_1(\mu,\nu^{(k)})=\frac{1}{4k}$.
    Moreover, since $v_i\le v_{i+1}$ and $v_j\in[0,1]$, we have that $(v_{i+1}-v_i)^2\le v_{i+1}-v_i$, and obtain 
    \[
    \lim_{n\to\infty}B^{(n)}_{ar}(\PMp)\le 4k\sum_{i=1}^k\Big[ \frac{(v_{i+1}-v_{i-1})}{2} \Big]
    \le 4k.
    \]
    That is, when the agents are distributed according to an uniform distribution,  the Bayesian approximation ratio of any percentile mechanism for the $k$-FLP is upper bounded by $4k$.
\end{example}
We now characterize the convergence rate of the Bayesian approximation ratio.
To do so, $\mu$ must have compact support or there exists $\delta>0$ such that
\begin{equation}
\label{eq:2_delta_mom}
    \int_{\erre}|x|^{2+\delta}d\mu <+\infty.
\end{equation}
In both cases, we have that the convergence rate is at most of the order $n^{-\frac{1}{2}}$.
\begin{theorem} 
% (Convergence rate)
\label{thm:conv_rate}
    Under the hypothesis of Theorem \ref{thm:limit}, let us further assume that either $\mu$ is supported on a compact set or satisfies \eqref{eq:2_delta_mom}.
    Then, we have that
    \begin{equation}
        \Bigg| \frac{\EE[\SCo_{\vec v}(\vec X_n)]}{\EE[\SCo_{opt}(\vec X_n)]}-\frac{W_1(\mu,\nu_{Q_{\vec v}})}{W_1(\mu,\nu^{(k)})} \Bigg| \le O(n^{-\frac{1}{2}}).
    \end{equation}
    Thus the convergence rate of the Bayesian approximation ratio of $\PMp$ is $O(n^{-\frac{1}{2}})$.
    Moreover, for every $\vec v\in(0,1)^k$, there exists $C>0$ such that
    \begin{equation}
    \label{eq:Barbound}
        B_{ar}^{(n)}(\PMp)\le \frac{W_1(\mu,\nu_{Q_{\vec v}})}{W_1(\mu,\nu^{(k)})} +\frac{C}{\sqrt{n}}\quad\quad \forall n>k,
    \end{equation}
    where $\nu_{Q_{\vec v}}$ is defined in \eqref{eq:nu_qp} and $\nu^{(k)}$ is a minimizer of \eqref{eq:projection_real}.
\end{theorem}

\begin{proof}
    It suffice to show that both $\big|\EE[\SCo_{opt}(\vec X_n)]-W_1(\mu,\nu^{(k)})\big|\le O(n^{-\frac{1}{2}})$ and $\big|\EE[\SCo_{\vec v}(\vec X_n)]-W_1(\mu,\nu_{Q_{\vec v}})\big|\le O(n^{-\frac{1}{2}})$ hold.
    Indeed, in this case, we have that
    \begin{align*}
        \Bigg|\frac{\EE[\SCo_{\vec v}(\vec X_n)]}{\EE[\SCo_{opt}(\vec X_n)]}&-\frac{W_1(\mu,\nu_{Q_{\vec v}})}{W_1(\mu,\nu^{(k)})}\Bigg|\\
        &=\Bigg|\frac{\EE[\SCo_{\vec v}(\vec X_n)]W_1(\mu,\nu^{(k)})-W_1(\mu,\nu_{Q_{\vec v}})\EE[\SCo_{opt}(\vec X_n)]}{\EE[\SCo_{opt}(\vec X_n)]W_1(\mu,\nu^{(k)})}\Bigg|\\
        &\le \Bigg|\frac{\EE[\SCo_{\vec v}(\vec X_n)]W_1(\mu,\nu^{(k)})-W_1(\mu,\nu_{Q_{\vec v}})W_1(\mu,\nu^{(k)})]}{\EE[\SCo_{opt}(\vec X_n)]W_1(\mu,\nu^{(k)})}\Bigg| \\
        &\quad+\Bigg|\frac{W_1(\mu,\nu_{Q_{\vec v}})W_1(\mu,\nu^{(k)})-W_1(\mu,\nu_{Q_{\vec v}})\EE[\SCo_{opt}(\vec X_n)]}{\EE[\SCo_{opt}(\vec X_n)]W_1(\mu,\nu^{(k)})}\Bigg|\\
        &\le \frac{\Big|\EE[\SCo_{\vec v}(\vec X_n)]-W_1(\mu,\nu_{Q_{\vec v}})\Big|}{\EE[\SCo_{opt}(\vec X_n)]}\\
        &\quad+\frac{W_1(\mu,\nu_{Q_{\vec v}})}{\EE[\SCo_{opt}(\vec X_n)]}\frac{\big|\EE[\SCo_{\vec v}(\vec X_n)]-W_1(\mu,\nu_{Q_{\vec v}})\big|}{W_1(\mu,\nu^{(k)})}\\
        &\le O(n^{-\frac{1}{2}}),
    \end{align*}
which concludes the proof.
The inequality $\big|\EE[\SCo_{\vec v}(\vec X_n)]-W_1(\mu,\nu^{(k)})\big|\le O(n^{-\frac{1}{2}})$ follows from the result in \cite{bobkov2019one}, since
\[
\big|\EE[\SCo_{\vec v}(\vec X_n)]-W_1(\mu,\nu^{(k)})\big|\le \EE[|\SCo_{\vec v}(\vec X_n)-W_1(\mu,\nu^{(k)})|]\le \EE[W_1(\mu_{\vec x},\mu)]\le O(n^{-\frac{1}{2}}).
\]
The identity $\EE\Big[\big|\SCo_{\vec v}(\vec X_n)-W_1(\mu,\nu_{Q_{\vec v}})\big|\Big]\le O(n^{-\frac{1}{2}})$ has been partially shown in the proof of Lemma \ref{lmm:lemma2}.
Indeed, we have shown that
\[
\Big|\EE[\SCo_{\vec v}(\vec x)]-W_1(\mu,\nu_{Q_{\vec v}})\Big| \le \EE\Big[\big|\SCo_{\vec v}(\vec X_n)-W_1(\mu,\nu_{Q_{\vec v}})\big|\Big]\le \EE[W_1(\mu_{\vec x},\mu)]+O(n^{-\frac{1}{2}}),
\]
which, in conjunction with the estimate $\EE[W_1(\mu_{\vec x},\mu)]\le O(n^{-\frac{1}{2}})$.

Lastly, we notice that, when $n>k$, we have that $\EE[SC_{opt}(\vec X_n)]>0$, hence $B_{ar}^{(n)}(\PMp)$ is well-defined.
By definition of $O(n^{-\frac{1}{2}})$, there exists a constant $C>0$ such that
\[
    \Big|B_{ar}^{(n)}(\PMp)-\frac{W_1(\mu,\nu_{Q_{\vec v}})}{W_1(\mu,\nu^{(k)})}\Big|\le \frac{C}{\sqrt{n}},
\]
thus
\[
B_{ar}^{(n)}(\PMp)\le \frac{W_1(\mu,\nu_{Q_{\vec v}})}{W_1(\mu,\nu^{(k)})} + \frac{C}{\sqrt{n}},
\]
which concludes the proof.
\qed\end{proof}

Notice that the term $\frac{W_1(\mu,\nu_{Q_{\vec v}})}{W_1(\mu,\nu^{(k)})}$ in \eqref{eq:Barbound}, is a constant that does not depend on $n$, but depends only on the specifics of the problem, that is $\mu$, $k$, and $\vec v$.

\begin{remark}
To conclude the section, we discuss the case in which $v_j\in\{0,1\}$ for at least one index $j\in [m]$.
For the sake of argument, let us consider a percentile mechanism induced by a percentile vector $\vec v$ such that $v_1=0$.
In this case, the mechanism places a facility at the position of the leftmost agent.
The asymptotic behaviour of the mechanism then depends on whether the support of $\mu$ is bounded from left or not.
Indeed, if $-\infty<a:=\inf_{x\in spt(\mu)}x$, we have that the position of the leftmost agent converges to $a$.
In this case, we can study the limit Bayesian approximation ratio of the percentile mechanism, but we will not be able to retrieve any guarantee on the convergence speed.
If $\inf_{x\in spt(\mu)}x=-\infty$, the position of the leftmost agent does not converge, thus we cannot adapt Theorem \ref{thm:limit} to suit this case.
% 
% {\color{red} We refer the reader to Appendix \ref{app:B} for a more detailed analysis of how to handle these cases.}
% 
\end{remark}

\section{The Optimal Percentile Mechanism}

Owing to Theorem \ref{thm:limit}, if $W_1(\mu,\nu_{Q_{\vec v}})=W_1(\mu,\nu^{(k)})$ the Bayesian approximation ratio of $\PMp$ converges to $1$ when $n\to \infty$.
We now show that, for any $k\in \mathbb{N}$ and any underlying distribution $\mu$, there exists a percentile vector whose associated mechanism asymptotically behaves optimally, \textit{i.e.} the limit of the Bayesian approximation ratio of the induced mechanism is equal to $1$.
Given an underlying distribution $\mu$, we denote with $\vec v_\mu$ its related optimal percentile vector.

\begin{theorem}
\label{crll:opt_vec}
Let $\mu$ be the underlying distribution and $\{y_j\}_{j\in [k]}$ be the support of the solution to problem \eqref{eq:projection_real}.
Then, the vector $\vec v_\mu$ defined as 
\begin{equation}
\label{eq:opt_per_vector_2}
    (v_\mu)_j=F_\mu(y_j), 
\end{equation}
is an optimal percentile vector.
% 
% Similarly, if $\{y_j\}_{j\in [k]}$ is the support of the solution to the projection problem \eqref{eq:projection_real}, formula \eqref{eq:opt_per_vector_2} defines the optimal percentile vector with respect to theSocial Cost.
\end{theorem}

\begin{proof}
By Lemma \ref{lmm:lemma2}, we have that
\[
\EE[\SCo_{\vec v_\mu}(\vec X_n)]\to W_1(\mu,\nu_{Q_{\vec v_\mu}})
\]
when $n \to \infty$.
We observe that $spt(\nu_{Q_{\vec v_\mu}})=spt(\nu^{(k)})$.
By definition of $\nu^{(k)}$, we have that
\[
W_1(\mu,\nu^{(k)})\le W_1(\mu,\nu_{Q_{\vec v_\mu}}).
\]
Furthermore, due to equation \eqref{eq:nuQclsoest}, we have
\[
W_1(\mu,\nu_{Q_{\vec v_\mu}})=\int_{-\infty}^{+\infty}\min_{j\in[k]}|x-y_j|d\mu=\min_{\{\lambda_j\}_{j\in [k]}}W_1(\lambda,\mu)\le W_1(\nu^{(k)},\mu),
\]
where $\lambda=\sum_{j=1}^k\lambda_j\delta_{y_j}$ and $y_j$ are the points in the support of both $\nu^{(k)}$ and $\nu_{Q_{\vec v_\mu}}$.
We then infer that 
\begin{equation}
    \label{eq:dimadditional}
    W_1(\mu,\nu^{(k)})=W_1(\mu,\nu_{Q_{\vec v_\mu}}),
\end{equation}
which concludes the proof.
\qed\end{proof}

\begin{remark}
    % 
% Although showing that $W_1(\mu,\nu_{Q_{\vec v_\mu}})=W_1(\mu,\nu^{(k)})$ is enough to prove Theorem \ref{crll:opt_vec}, i
% 
It is also worth of notice that $\nu_{Q_{\vec v_\mu}}=\nu^{(k)}$ holds.
Indeed, toward a contradiction, let us assume that $\nu_{Q_{\vec v_\mu}}\neq \nu^{(k)}$.
Then there exists a $\bar j\in[k]$ such that $(\nu_{Q_{\vec v_\mu}})_i=\nu^{(k)}_i$ for every $i=1,\dots,\bar j-1$ and $(\nu_{Q_{\vec v_\mu}})_{\bar j} \neq \nu^{(k)}_{\bar j}$.
Since the optimal transportation plan between two measures supported over a line is monotone, we have that $W_1(\mu,\nu^{(k)})=\sum_{j=0}^{k}\int_{l_{j}}^{l_{j+1}}|x-y_j|d\mu$,
where $l_0=-\infty$ and $l_r=F_\mu^{[-1]}(\sum_{i=1}^r\nu^{(k)}_i)$ for every $r\in[k]$.
Since $(\nu_{Q_{\vec v_\mu}})_{\bar j} \neq \nu^{(k)}_{\bar j}$, we have that $l_{\bar j}\neq \frac{y_{\bar j}+y_{\bar j +1}}{2}$.
Thus we have $W_1(\mu,\nu^{(k)})\neq \int_{-\infty}^{+\infty}\min_{j\in[k]}|x-y_j|d\mu$, which contradicts the definition of $\nu^{(k)}$ and \eqref{eq:nuQclsoest}, thus $\nu_{Q_{\vec v_\mu}}=\nu^{(k)}$.
\end{remark}

Given $k\in\enne$ and a probability measure $\mu$, it is possible to retrieve a system of $k$ equations that characterizes the optimal percentile mechanism. 
Indeed, let us denote with $y_1,\dots,y_k$ the support of the solution to $\min_{\lambda\in\PP_k(\erre)}W_1(\mu,\lambda)$ and let $z_i=\frac{y_i+y_{i+1}}{2}$ for $i=1,\dots,k-1$, $z_0=-\infty$, and $z_k=+\infty$.
% 
% Let us now focus on $y_1$.
% 
Since every agent's cost is defined by its distance to the closest facility, we know that every agent in $(z_0,z_1)$ will access the facility located in $y_1$.
Due to the optimality of the solution, we infer that $y_1$ is locally optimal over the set $(z_0,z_1)$.
Otherwise, we could reduce the cost of the solution by replacing $y_1$ with the optimal facility location for the problem restricted to $(z_0,z_1)$.
Since we are considering the Social Cost, the local optimality of $y_1$ is expressed by the identity $2(F_\mu(y_1)-F_\mu(z_0))=F(z_1)-F_\mu(z_0)$, since $y_1$ has to be the median of $\mu$ when the measure is restricted to $(z_0,z_1)$.
% 

% 
% By expressing the local optimality condition for every facility location and replacing every $z_i$ with its definition, we get the following theorem.

\begin{theorem}
\label{thm:system_charact}
    Given $k\in\enne$ and $\mu\in\PP(\erre)$, let $\nu$ be a solution to Problem \eqref{eq:projection_real}.
    Then the optimal percentile vector is $\vec v_\mu=(F_\mu(y_1),\dots,F_\mu(y_m))\in(0,1)^k$, where $y_1\le y_2\le \dots\le y_k$ satisfy the following system of $k$ equations
\begin{equation}
\label{eq:system_gen2_app}
    \begin{cases}
        2F_\mu(y_1)=F_\mu\Big(\frac{y_1+y_2}{2}\Big)\\
        % 2\bigg(F_\mu(y_2)-F_\mu\Big(\frac{y_1+y_2}{2}\Big)\bigg)=F_\mu\Big(\frac{y_2+y_3}{2}\Big)-F_\mu\Big(\frac{y_1+y_2}{2}\Big)\\
        % \quad\quad\quad\dots\\
        2\bigg(F_\mu(y_{j-1})-F_\mu\Big(\frac{y_{j-2}+y_{j-1}}{2}\Big)\bigg)=F_\mu\Big(\frac{y_{j}+y_{j-1}}{2}\Big)-F_\mu\Big(\frac{y_{j-1}+y_{j-2}}{2}\Big),\;\; j\in[k-1].\\
        2\bigg(F_\mu(y_k)-F_\mu\Big(\frac{y_k+y_{k-1}}{2}\Big)\bigg)=1-F_\mu\Big(\frac{y_k+y_{k-1}}{2}\Big)\\
    \end{cases}
\end{equation}
 % are the points in the support of $\nu$.
\end{theorem}

\begin{proof}
    Given $k\in\enne$ and a probability measure $\mu$,  let us denote with $y_1,\dots,y_k$ the support of the solution to $\min_{\lambda\in\PP_k(\erre)}W_1(\mu,\lambda)$ and let $z_i=\frac{y_i+y_{i+1}}{2}$ for $i=1,\dots,k-1$, $z_0=-\infty$, and $z_k=+\infty$.
    Let us now focus on $y_1$.
    Since every agent's cost is defined by its distance to the closest facility, we know that every agent in $(z_0,z_1)$ will access the facility located in $y_1$.
    Otherwise, we could reduce the cost of the solution by assigning an agent to a closer facility.
    Moreover, due to the optimality of the solution, we infer that $y_1$ is locally optimal over the set $(z_0,z_1)$, thus $y_1$ is the median of $\mu$ restricted to $(z_0,z_1)$.
    Otherwise, we could reduce the cost of the solution by replacing $y_1$ with the optimal facility location for the problem restricted to $(z_0,z_1)$.
    Therefore, we have that $2(F_\mu(y_1)-F_\mu(z_0))=F(z_1)-F_\mu(z_0)$, since $y_1$ has to be the median of $\mu$ when the measure is restricted to $(z_0,z_1)$.
\qed\end{proof}

Since the projection problem \eqref{eq:projection_real} admits a solution, the system \eqref{eq:system_gen2_app} always admits at least a solution. 
In Appendix \ref{app:C}, we use Theorem \ref{thm:system_charact} to compute the optimal percentile mechanism for the uniform, normal, and exponential distributions.
In Table \ref{tab:5_intro}, we summarize our finding.

\begin{table}[t!]
\centering
\begin{tabular}{l@{\hskip 0.25in}  c@{\hskip 0.5in} c@{\hskip 0.5in} c}
\toprule
% & \multicolumn{3}{c|}{$\SCo$} & \multicolumn{3}{c}{$C_{l_2}$} \\
      % \cline{2-4}
      &  {$k=1$} & {$k=2$} & {$k=3$}   \\
      \midrule
      {$\mathcal{N}$} & $(0.5)$ & $(0.25,0.75)$ & $(0.15,0.5,0.85)$ \\
      {$\mathcal{E}$} & $(0.5)$ & $(0.33,0.83)$ & $(0.25,0.67,0.92)$ \\
      {$\mathcal{U}$}  & $(0.5)$ & $(0.25,0.75)$ & $(0.16,0.5,0.83)$ \\
\bottomrule
\end{tabular}
% 
% \begin{tabular}{c  |   c|c|c   |   c|c|c}
% \toprule
% & \multicolumn{3}{c|}{$\SCo$} & \multicolumn{3}{c}{$C_{l_2}$} \\
%       \cline{2-7}
%       &  {$k=1$} & {$k=2$} & {$k=3$}  &  {$k=1$} & {$k=2$} & {$k=3$}  \\
%       \cline{1-7}
%       {$\mathcal{N}$} & $(0.5)$ & $(0.25,0.75)$ & $(0.15,0.5,0.85)$ & $(0.5)$ & $(0.16,0.84)$ & $(0.17,0.5,0.83)$ \\
%       {$\mathcal{E}$} & $(0.5)$ & $(0.33,0.83)$ & $(0.25,0.67,0.92)$ & $(1-e^{-1})$ & $(0.45,0.92)$ & $(0.35,0.8,0.97)$ \\
%       {$\mathcal{U}$}  & $(0.5)$ & $(0.25,0.75)$ & $(0.16,0.5,0.83)$ & $(0.5)$ & $(0.25,0.75)$ & $(0.16,0.5,0.83)$ \\
% \bottomrule
% \end{tabular}
  \caption{\small The asymptotically optimal percentile vectors for the Normal ($\mathcal{N}$), Exponential ($\mathcal{E}$), and Uniform distribution ($\mathcal{U}$). Every row contains the optimal percentile vectors of a distribution for $1$, $2$, and $3$ facilities and with respects to the Social Cost. \label{tab:5_intro} } 
\end{table}
% 

% 
% Owing to the hypothesis on $\mu$ and to Theorem \ref{thm:estimation_perc_vect}, we infer that the optimal percentile mechanism does not have any $0$-$1$ entry, and thus infer the following theorem.
% 
\begin{theorem}
\label{thm:grs}
 Given a probability distribution $\mu$ such that condition \eqref{eq:finite_moment} is satisfied, let $\vec v_\mu$ be the optimal percentile vector associated with $\mu$.
 Then, there exists a constant $C$ such that, for every $n>k$, we have  $B_{ar}^{(n)}(\mathcal{PM}_{\vec v_{\mu}})\le 1+\frac{C}{\sqrt{n}}$.
\end{theorem}

\begin{proof}
    From Theorem \ref{thm:system_charact}, we have that the optimal percentile vector $\vec v_{\mu}$ belongs to $(0,1)^k$.
    In particular, we have that $\vec v_{\mu}$ meets the conditions of Theorem \ref{thm:limit} and \ref{thm:conv_rate}, hence there exists a constant $C>0$ such that
    \[
        B_{ar}^{(n)}(\mathcal{PM}_{\vec v_{\mu}})\le \lim_{n\to\infty}B_{ar}^{(n)}(\mathcal{PM}_{\vec v_{\mu}})+\frac{C}{\sqrt{n}}.
    \]
    To conclude the proof, it suffices to notice that, by definition of optimal percentile vector, we have that $\lim_{n\to\infty}B_{ar}^{(n)}(\mathcal{PM}_{\vec v_{\mu}})=1$.
\qed\end{proof}

% \subsubsection{The Scale Invariance of the Optimal Percentile Vector.}
% 
To conclude, we show that the limit of the Bayesian approximation ratio of the percentile mechanisms and the percentile vector $\vec v_\mu$ defined in \eqref{eq:opt_per_vector_2} are immune to scale changes.
This is particularly useful when the mechanism designer only knows the class of distribution to which the agents' distribution belongs.
For example, the designer might know that the agents' type follows a Gaussian distribution but is unaware of its mean and/or its standard deviation.
In the following, we show that the optimal percentile and the limit of the Bayesian approximation ratio of the percentile mechanisms are the same regardless of the mean or standard deviation of the distribution.

\begin{theorem}
\label{thm:scale_invariance}
Let $X$ be the random variable that describes the agents' type distribution.
If $\vec v_\mu$ is the optimal percentile vector associated with $X$, then $\vec v_\mu$ is also the optimal percentile vector for any random variable of the form $X':=\sigma X + m$, where $m\in \erre$ and $\sigma> 0$.
% 
% In particular, the approximation ratio of any percentile mechanism $\PMp$ does not depend on the mean nor the standard deviation of the distribution $\mu$.
% 
% This holds for the Social, the Maximum, and the $l_p$ Costs.
\end{theorem}

\begin{proof}
    Let $L(x)=\sigma x +m$ be a linear function such that $\sigma>0$.
    Since $\sigma>0$, $L$ is a bijective and monotone-increasing function.
    We denote the inverse function of $L$ with $H$.
    It is well-known that if $X$ is a random variable whose law is $\mu$, then the law of $L(X)$ is $L_\#\mu$.
    \footnote{We recall that $L_\#\mu$ is the pushforward of the measure $\mu$, defined as $L_\#\mu(A)=\mu(L^{-1}(A))$, see \cite{villani2009optimal}.}
    We now show that if $\nu^{(k)}$ is the projection of $\mu$ over $\PP_k(\erre)$, then $L_\#\nu^{(k)}$ is a projection of $L_\# \mu$ over $\PP_k(\erre)$.
    First, notice that, since $\nu^{(k)}\in\PP_k(\erre)$, then $L_\#\nu^{(k)}\in\PP_k(\erre)$.
    Toward a contradiction, let $\gamma$ be the projection of $L_\#\mu$ over $\PP_k(\erre)$.
    By definition of $\gamma$, we have
    \[
        W_1(\gamma,L_\#\mu)<W_1(L_\#\nu^{(k)},L_\#\mu).
    \]
    Let us now define $\eta=H_\#\gamma\in\PP_k(\erre)$, by definition we have $\gamma=L_\#\eta$.
    Furthermore, by the properties of the Wasserstein distance, we have
    \begin{equation}
    \label{eq:scalingWopt}
        W_1(L_\#\eta,L_\#\mu)=\sigma W_1(\eta,\mu)
    \end{equation}
    and 
    \[
        W_1(L_\#\nu^{(k)},L_\#\mu)=\sigma W_1(\nu^{(k)},\mu).
    \]
    In particular, we get
    \[
        W_1(\eta,\mu)<W_1(\nu^{(k)},\mu),
    \]
    which contradicts the optimality of $\nu^{(k)}$.
    We therefore conclude that $L_\#\nu^{(k)}$ is a projection of $L_\#\mu$ over $\PP_k(\erre)$.
    Notice that $y\in spt(\nu^{(k)})$ if and only if $\sigma y + m\in spt(L_\#\nu^{(k)})$.
    From Corollary \ref{crll:opt_vec}, we have that the $j$-th entry of the optimal vector related to $\mu$ is $F_\mu(y_j)$.
    Similarly, the $j$-th entry of the optimal vector related to $L_\#\mu$ is $F_{L_\#\mu}(\sigma y_j+m)=F_\mu(y_j)$, thus, by Theorem \ref{crll:opt_vec}, the optimal percentile vector with respect to $\sigma X + m$ is also optimal with respect to $X$ and vice-versa.
    Finally, we notice that the Bayesian approximation ratio of every $\PMp$ is the same for both $\mu$ and $L_\#\mu$.
    It follows from \eqref{eq:scalingWopt} and $W_1(L_\#\mu,L_\#\nu_{Q_{\vec v}})=\sigma W_1(\mu,\nu_{Q_{\vec v}})$.
    % 

    % 
    % Through a similar argument, we infer the same conclusion for every $l_p$ Cost and the Maximum Cost.
\qed\end{proof}

Theorem \ref{thm:scale_invariance} formalizes the following observation: the optimal facility locations and the output of any percentile mechanism do not depend on the scale.
Indeed, given a percentile vector $\vec v$ and the number of agents $n$, if the agents' positions are sampled from a random variable $X$, the output of $\PMp$ is the vector containing the $(\floor{(n-1)v_j}+1)$-th order statistics of the sample.
Since the ordering of the values is unaffected by positive affine transformations, scaling any sample just magnifies (or shrinks) the cost of the output according to $\sigma$.
Similarly, if we scale the agents' positions, the optimal facility locations will scale accordingly.
Hence the ratio of the two costs is immune to scale changes.
% 
% Finally, we notice that if we allow $\sigma$ to be negative, Theorem \ref{thm:scale_invariance} is no longer valid.
% 
% Indeed, multiplying a generic random variable by a negative quantity alters its quantiles and thus the optimal percentile vector.

\section{Computing the Optimal Percentile Mechanism from an approximation of \texorpdfstring{$\mu$}{mu}.}
\label{sec5}

In the previous section, we have shown that if the agents' types are sampled from a common probability distribution $\mu$, the mechanism designer is able to detect a percentile mechanism whose cost is asymptotically optimal when the number of agents increases.
To detect the optimal percentile vector, however, it is necessary to have access to the distribution $\mu$.
In many cases, this is not feasible since the designer has only access to an approximation or a prediction of agents' distribution, which we denote with $\Tilde{\mu}$.
Thus, the designer is able to compute $\vec v_{\tilde\mu}$ rather than the real optimal percentile vector $\vec v_\mu$.
We now show that it is possible to estimate the difference between the limit of the Bayesian approximation ratio of $\mathcal{PM}_{\vec v_{\tilde\mu}}$ and $1$, i.e. the limit of the Bayesian approximation ratio of the optimal percentile mechanism.
In particular, we show that the closer $\mu$ and $\Tilde{\mu}$ are with respect to the $W_\infty$ distance, the closer the asymptotic cost of $\mathcal{PM}_{\vec v_{\tilde \mu}}$ gets to the optimal cost.

\begin{theorem}
\label{thm:error}
Let $\Tilde{\mu}$ and $\mu$ be two probability measures supported over a compact interval $I$.
% and 
Let $\vec v_{\tilde\mu}$ be the percentile vector obtained by solving the system \eqref{eq:opt_per_vector_2} by using $\Tilde{\mu}$ instead of $\mu$.
Then, we have
\begin{equation}
    \lim_{n\to \infty}\Bigg|\frac{\EE[\SCo_{\vec v_{\tilde\mu}}(\vec X_n)]}{\EE[\SCo_{opt}(\vec X_n)]}-1\Bigg|\le \frac{W_\infty(\mu,\Tilde{\mu})+2W_1(\mu,\Tilde{\mu})}{W_1(\mu,\nu^{(k)})}\le 3\frac{W_\infty(\mu,\Tilde{\mu})}{W_1(\mu,\nu^{(k)})},
\end{equation}
where $\EE$ is the expected value with respect to the real agents' distribution $\mu$.
In particular, for every $n>k$, there exists a constant $C>0$ such that  
\[
|B_{ar}^{(n)}(\mathcal{PM}_{\vec v_{\mu}})-B_{ar}^{(n)}(\mathcal{PM}_{\vec v_{\tilde \mu}})|\le 3\frac{W_\infty(\mu,\Tilde{\mu})}{W_1(\mu,\nu^{(k)})} + \frac{C}{\sqrt{n}},
\]
where $\vec v_{\mu}$ is the optimal percentile vector associated with $\mu$. 
\end{theorem}

\begin{proof}
    Let us denote with $\Tilde{\nu}^{(k)}$ the projection of $\Tilde{\mu}$ over $\PP_k(\erre)$ and with $\nu^{(k)}$ the projection of $\mu$ over $\PP_k(\erre)$.
    We denote with $\{y_j\}_{j\in [k]}$ the support of $\nu^{(k)}$ and with $\{\Tilde{y}_j\}_{j\in [k]}$ the support of $\Tilde{\nu}^{(k)}$.
    Accordingly, we denote with $\vec v_{\Tilde{\mu}}$ and $\vec v_\mu$ the optimal percentile vectors associated with $\Tilde{\mu}$ and $\mu$, respectively.
    By Lemma \ref{lmm:lemma2}, we have that the numerator of the Bayesian approximation ratio converges to $W_1(\mu,\nu_{Q_{ \vec v_{\Tilde{\mu}}}})$, where $\nu_{Q_{ \vec v_{\Tilde{\mu}}}}$ is defined as in \eqref{eq:nu_qp}.
    Let us now consider, $\beta_{\vec v_{\Tilde{\mu}}}$ defined  as $\beta_{\vec v_{\Tilde{\mu}}}:=\sum_{j\in [k]}(\Tilde{\nu}^{(k)})_{j}\delta_{z_j}$,
    % \[
    %     \beta_{\vec v_{\Tilde{\mu}}}:=\sum_{j\in [k]}(\Tilde{\nu}^{(k)})_{j}\delta_{z_j},
    % \]
    where $\vec z=(z_1,\dots,z_k)$ is the support of $\nu_{Q_{\vec v_{\Tilde{\mu}}}}$, then
    % 
    % Then we have
    \begin{align}
    \label{eq:chain_errors}
        \nonumber W_1(\mu,\nu_{Q_{ \vec v_{\Tilde{\mu}}}})&\le W_1(\mu,\beta_{\vec v_{\Tilde{\mu}}})\le W_1(\mu,\Tilde{\mu}) + W_1(\Tilde{\mu},\beta_{\vec v_{\Tilde{\mu}}})\\
        &\le W_1(\mu,\Tilde{\mu})+ W_1(\Tilde{\mu},\Tilde{\nu}^{(k)})+W_1(\Tilde{\nu}^{(k)},\beta_{\vec v_{\Tilde{\mu}}})\\
        \nonumber &\le   W_1(\mu,\Tilde{\mu})+ W_1(\Tilde{\mu},\nu^{(k)})+W_1(\Tilde{\nu}^{(k)},\beta_{\vec v_{\Tilde{\mu}}})\\
       \nonumber &\le 2W_1(\mu,\Tilde{\mu})+ W_1(\mu,\nu^{(k)})+W_1(\Tilde{\nu}^{(k)},\beta_{\vec v_{\Tilde{\mu}}}).
    \end{align}
    By definition of $\beta_{\vec v_{\Tilde{\mu}}}$ and $\Tilde{\nu}^{(k)}$, we have $W_1(\Tilde{\nu}^{(k)},\beta_{\vec v_{\Tilde{\mu}}})\le\sum_{j\in [k]}(\Tilde{\nu}^{(k)})_{j}|F^{[-1]}_{\Tilde{\mu}}(\Tilde{p}_j)-F^{[-1]}_{\mu}(\Tilde{p}_j)|$.
    % \[
    %     W_1(\Tilde{\nu}^{(k)},\beta_{\vec v_{\Tilde{\mu}}})\le\sum_{j\in [k]}(\Tilde{\nu}^{(k)})_{j}|F^{[-1]}_{\Tilde{\mu}}(\Tilde{p}_j)-F^{[-1]}_{\mu}(\Tilde{p}_j)|.
    % \]
    Since $W_1(\mu,\tilde \mu)\le W_\infty(\mu,\tilde \mu)$ and $W_\infty(\mu,\Tilde{\mu})=\max_{\ell\in [0,1]}|F^{[-1]}_{\mu}(\ell)-F^{[-1]}_{\Tilde{\mu}}(\ell)|$, we infer $W_1(\Tilde{\nu}^{(k)},\beta_{\vec v_{\Tilde{\mu}}})\le W_\infty(\mu,\Tilde{\mu})$.
    To conclude, we notice that, 
    % from Theorem \ref{thm:estimation_perc_vect} we infer that the optimal percentile vectors $\vec v_{\mu}$ and $\vec v_{\tilde{\mu}}$ belong to the set $(0,1)^{k}$.
    % 
    % Then we have
    \begin{align*}
        |B_{ar}^{(n)}&(\mathcal{PM}_{\vec v_{\mu}})-B_{ar}^{(n)}(\mathcal{PM}_{\vec v_{\tilde \mu}})|\\
        &\le |1-B_{ar}(\mathcal{PM}_{\vec v_{\tilde \mu}})|+|B_{ar}^{(n)}(\mathcal{PM}_{\vec v_{\mu}})-1|+|B_{ar}(\mathcal{PM}_{\vec v_{\tilde \mu}})-B_{ar}^{(n)}(\mathcal{PM}_{\vec v_{\tilde \mu}})|\\
        &\le 3\frac{W_\infty(\mu,\Tilde{\mu})}{W_1(\mu,\nu^{(k)})}+ |B_{ar}^{(n)}(\mathcal{PM}_{\vec v_{\mu}})-1|+ |B_{ar}(\mathcal{PM}_{\vec v_{\tilde \mu}})-B_{ar}^{(n)}(\mathcal{PM}_{\vec v_{\tilde \mu}})|,
    \end{align*}
    where $B_{ar}(\mathcal{PM}_{\vec v_{\tilde \mu}})=\lim_{n\to\infty}B_{ar}^{(n)}(\mathcal{PM}_{\vec v_{\tilde \mu}})$.
    Since $\mu$ has compact support and $\vec v_\mu,\vec v_{\tilde \mu}\in (0,1)^k$, we infer that $|B_{ar}^{(n)}(\mathcal{PM}_{\vec v_{\mu}})-1| \le O(n^{-\frac{1}{2}})$ and $|B_{ar}(\mathcal{PM}_{\vec v_{\tilde \mu}})-B_{ar}^{(n)}(\mathcal{PM}_{\vec v_{\tilde \mu}})|\le O(n^{-\frac{1}{2}})$.
    Thus there exist a constant $C>0$ such that
    \[
    |B_{ar}^{(n)}(\mathcal{PM}_{\vec v_{\mu}})-B_{ar}^{(n)}(\mathcal{PM}_{\vec v_{\tilde \mu}})|\le 3\frac{W_\infty(\mu,\Tilde{\mu})}{W_1(\mu,\nu^{(k)})} + \frac{C}{\sqrt{n}}
    \]
    which concludes the proof.
\qed\end{proof}

\section{Conclusion and Future Works}
\label{sec:conclusion}

In this paper, we studied the percentile mechanisms in the Bayesian Mechanism Design framework.
We have shown that the ratio between the expected cost of the mechanisms and the expected optimal cost converges to a constant as the number of agents goes to infinity.
We have characterized both the limit value and the convergence speed.
We then showed that for every underlying distribution $\mu$, there exists an optimal percentile vector $\vec{v}_\mu$ that does not depend on the mean or the variance of the distribution.
The scale invariance property allows us to compute the optimal percentile vector when the designer only knows the class to which the probability measure belongs.
Lastly, we have shown that determining the optimal percentile mechanism from an approximation of the underlying distribution leads to a mechanism whose performance is quasi-optimal as long as the approximation is close to the real distribution with respect to $W_\infty$.
A natural open question is whether our formalism could be adopted to extend our results to higher dimensional cases.
In \cite{sui2013analysis}, the percentile mechanisms are generalized to higher dimensions by dealing with each dimension separately.
This suggests that our approach can be extended to also handle higher-dimensional problems thanks to the fact that also the Wasserstein Distance can be separate along each cardinal direction \cite{auricchio2024pythagorean,auricchio2018computing}.
Moreover, our framework can be extended beyond the classic $k$-FLP.
In particular, it is foreseeable to use our results to tackle the case in which agents have fractional preferences.
Another interesting direction is to adapt our reformulation of the problem through OT theory to design and study randomized mechanisms for the $k$-FLP.

\begin{credits}
\subsubsection{\ackname} Jie Zhang was partially supported by a Leverhulme Trust Research Project Grant (2021 -- 2024) and the EPSRC grant (EP/W014912/1).

% \subsubsection{\discintname}
% It is now necessary to declare any competing interests or to specifically
% state that the authors have no competing interests. Please place the
% statement with a bold run-in heading in small font size beneath the
% (optional) acknowledgments\footnote{If EquinOCS, our proceedings submission
% system, is used, then the disclaimer can be provided directly in the system.},
% for example: The authors have no competing interests to declare that are
% relevant to the content of this article. Or: Author A has received research
% grants from Company W. Author B has received a speaker honorarium from
% Company X and owns stock in Company Y. Author C is a member of committee Z.
\end{credits}
%
% ---- Bibliography ----
%
% BibTeX users should specify bibliography style 'splncs04'.
% References will then be sorted and formatted in the correct style.

\bibliographystyle{splncs04}
\bibliography{bibliography}

\appendix

% \section{Appendix  }

% In this appendix, we report all the material missing from the main body of the paper.

\clearpage

\section{Optimal Percentile Vectors of classic probability distributions.}
% \label{sec:optimal_perc}
\label{app:C}
In this appendix, we study the optimal percentile vector with respect to the Social Cost.
In particular, we retrieve the formula of the optimal percentile mechanisms when the underlying distribution $\mu$ is Uniform, Exponential, or Gaussian, and we need to locate $1$, $2$, or $3$ facilities.
First, we focus on the cases $k=1$ and $k=2$.
We then deal with $k=3$ and show how to approach the general case $k>3$.

\paragraph{One Facility.}

When $k=1$, let $X$ be the random variable describing the agents' type, and $y$ be the optimal location of the facility. 
Let $y$ be the real value that minimizes $\EE[|X-y|]$.
It is well-known that the median mechanism is optimal with respect to the Social Cost function when we have a discrete input.
Similarly, the value $y$ that minimizes $\EE[|X-y|]$ is the median of $X$. 
Thus, as we are about to see, the median mechanism retains its optimality also in the Bayesian framework.

\begin{theorem}
\label{thm:one_fac}
    The median mechanism is optimal for the $1$-FLP.
\end{theorem}

\begin{proof}
    It follows from the fact that the median is the value that minimizes the function $c\to \EE[|X-c|]$.
\end{proof}

\paragraph{Two Facilities.}

When $k=2$, the symmetry of the distribution is sufficient to retrieve the optimal percentile mechanism with respect to the Social Cost.
Moreover, as for the one facility case, the optimal percentile vector does not depend on the measure $\mu$.

\begin{theorem}
\label{thm:sym_utilitarian}
    If $\mu$ is symmetric, the optimal percentile mechanism with respect to the Social Cost is induced by the vector $\vec v=(0.25,0.75)$.
\end{theorem}

\begin{proof}
    Since $\mu$ is symmetric, we have that the optimal locations for two facilities are symmetric with respect to the median (or, equivalently, the mean) of $\mu$.
    In particular, the median of $\mu$ is equidistant from both facilities.
    It then suffices to retrieve the location of the facility on the right to the median to retrieve the other one.
    Since the facility locations are globally optimal, every facility location has to be locally optimal, thus the facility on the left of the median minimizes the functional $\EE[|X-c|]$, where the expected value is taken with respect to $\mu$ restricted on the set $\{X\le med(\mu)\}$.
    By symmetry, the other facility will be located at the median of $\mu$ restricted over $\{X\ge med(\mu)\}$, from which we conclude that the optimal percentile vector is $\vec p =(0.25,0.75)$.
\end{proof}

For asymmetric distributions, the computation of the optimal percentile mechanism is less straightforward and requires solving a system of equations.
In the following, we compute the optimal vector for the exponential distribution.

\begin{theorem}
    \label{two_facilities_exp}
    Let $\mu$ be an exponential probability distribution.
    Then, the optimal percentile vector with respect to the Social Cost for two facilities is $\vec v_{\mathcal{E}}=(0.33,0.83)$.
\end{theorem}

\begin{proof}
Let $y_1$ and $y_2$ be the points in the support of the projection of $\mu$ over $\PP_2(\erre)$.
Let us denote with $z$ the middle point of $y_1$ and $y_2$, that is $z=\frac{1}{2}(y_1+y_2)$.
We then have that agents on the left of $z$ will be using the facility at $y_1$, while the ones on the right, will use the one in $y_2$.
In particular, the expected Social Cost is
\[
\int_{0}^{z}|x-y_1|e^{-x}dx+\int_{z}^{\infty}|x-y_2|e^{-x}dx.
\]
We now show that $y_1$ is the median of $\mu$ restricted to $\{x\le z\}$.
Assume, toward a contradiction, that $y_1$ is not the median of $\mu$ restricted over $\{x\le z\}$.
There exists a value $y_1'$ such that
\[
\int_{0}^{z}|x-y_1|e^{-x}dx >\int_{0}^{z}|x-y_1'|e^{-x}dx.
\]
Hence we get
\begin{align*}
    \int_{0}^{z}|x-y_1|e^{-x}dx+\int_{z}^{\infty}|x-y_2|e^{-x}dx&>\int_{0}^{z}|x-y_1'|e^{-x}dx+\int_{z}^{\infty}|x-y_2|e^{-x}dx\\
    &\ge \min_{\lambda_1,\lambda_2}W_1(\mu,\lambda),
\end{align*}
where $\lambda=\lambda_1\delta_{y_1'}+\lambda_2\delta_{y_2}$, which contradicts the optimality of the facility locations.
We therefore conclude that $y_1$ is the median of $\mu$ restricted over $(-\infty,z)$, i.e.
\[
2F_\mu(y_1)=F_\mu(z).
\]
Similarly, we get that
\[
2(1-F_\mu(y_2))=(1-F_\mu(z)).
\]
Since $F_\mu(t)=1-e^{-t}$, we get the following system of equations
\begin{equation}
\begin{cases}
2(1-e^{-y_1})=(1-e^{-z}) \\
2e^{-y_2}=e^{-z} \\ 
z=\frac{1}{2}(y_1+y_2)
\end{cases}.
\end{equation}
From the second and third equations, we retrieve $y_2=y_1+2\log(2)$, which, plugged in the first equation, allows us to conclude that $y_1=log(3)-\log(2)$ and $y_2=\log(6)$.
From equation \eqref{eq:opt_per_vector_2}, we conclude the proof.
\qed
\end{proof}

\paragraph{Three  and More Facilities.}

Lastly, we consider the case $k\ge 3$.
We show that the optimal percentile vector can be retrieved by solving a system of equations obtained by expressing the local optimality properties of the solution to problem \eqref{eq:projection_real}.
Let us now consider a probability distribution $\mu$ and its c.d.f., namely $F_\mu$.
Given $k$, let us denote with $y_1,\dots,y_k$ the support of the solution to
\[
\min_{\lambda\in\PP_k(\erre)}W_1(\mu,\lambda).
\]
Let us now define $z_i=\frac{y_i+y_{i+1}}{2}$ for $i=1,\dots,k-1$, $z_0=-\infty$, and $z_k=+\infty$.
Let us now focus on $y_1$.
Since every agent's cost is defined by its distance to the closest facility, we know that every agent in $(z_0,z_1)$ will access the facility located in $y_1$.
Due to the optimality of the solution, we infer that $y_1$ is locally optimal over the set $(z_0,z_1)$.
Otherwise, we could reduce the cost of the solution by replacing $y_1$ with the optimal facility location for the problem restricted to $(z_0,z_1)$.
Since we are considering the Social Cost, the local optimality of $y_1$ is expressed by the identity $2(F_\mu(y_1)-F_\mu(z_0))=F(z_1)-F_\mu(z_0)$, since $y_1$ has to be the median of $\mu$ when the measure is restricted to $(z_0,z_1)$.
By expressing the local optimality condition for every facility location and replacing every $z_i$ with its definition, we get the following system of $k$ equations
\begin{equation}
\label{eq:system_gen2}
    \begin{cases}
        2F_\mu(y_1)=F_\mu\Big(\frac{y_1+y_2}{2}\Big)\\
        2\bigg(F_\mu(y_2)-F_\mu\Big(\frac{y_1+y_2}{2}\Big)\bigg)=F_\mu\Big(\frac{y_2+y_3}{2}\Big)-F_\mu\Big(\frac{y_1+y_2}{2}\Big)\\
        \quad\quad\quad\dots\\
        2\bigg(F_\mu(y_{k-1})-F_\mu\Big(\frac{y_{k-2}+y_{k-1}}{2}\Big)\bigg)=F_\mu\Big(\frac{y_{k-1}+y_{k-2}}{2}\Big)-F_\mu\Big(\frac{y_{k-1}+y_{k-2}}{2}\Big)\\
        2\bigg(F_\mu(y_k)-F_\mu\Big(\frac{y_k+y_{k-1}}{2}\Big)\bigg)=1-F_\mu\Big(\frac{y_k+y_{k-1}}{2}\Big)\\
    \end{cases}.
\end{equation}
Notice that, since the projection problem \eqref{eq:projection_real} admits a solution, the system \eqref{eq:system_gen2} always admits at least a solution. 
To conclude this section, we leverage these characterizations to explicitly compute the optimal percentile vector for the $3$-FLP with respect to some classic probability distributions.
As for the previous cases, we distinguish between symmetric and asymmetric distributions, since by exploiting the symmetry of the measure it is possible to simplify the problem.

\begin{theorem}
\label{thm:fac_3_sym}
    Let $\mu$ be a symmetric distribution whose mean is equal to $0$.
    Then, the second entry of the optimal percentile mechanism is $0.5$.
    The other two entries are $F_\mu(\pm 2z)$, where $z$ minimizes the following function $z \to \int_{0}^z x^1 d\mu+\int_{z}^{+\infty}|x-2z| d\mu$.
    In particular, the optimal percentile vectors for the Uniform and Gaussian distribution are $\vec v_{\mathcal{U}}\sim(0.16,0.5,0.83)$ and $\vec v_\mathcal{N}\sim(0.15,0.5,0.85)$, respectively.
\end{theorem}

\begin{proof}
    By a symmetry argument, we have that one of the three facilities has to be located at the mean of the measure.
    Similarly, the other two facilities must be symmetrically located with respect to the mean.
    Let us now denote with $2y$ the position of the facility on the right of $0$.
    Therefore $y$ is the minimum of the function 
    \[
    z\to \int_{0}^{z}xd\mu+\int_{z}^{\infty}|x-2z|d\mu,
    \]
    which concludes the first half of the proof.
    Let us now consider the Gaussian distribution, since the other cases are similar.
    In this case, up to a constant, we need to minimize the following functional
    \begin{align*}
         z \to &\int_{0}^{z}xe^{-\frac{x^2}{2}}dx+\int_{z}^{\infty}|x-2z|e^{-\frac{x^2}{2}} dx\\
         =&\int_{0}^{z}xe^{-\frac{x^2}{2}}dx+\int_{z}^{2z}(2z-x)e^{-\frac{x^2}{2}} dx +\int_{2z}^{\infty}(x-2z)e^{-\frac{x^2}{2}} dx
    \end{align*}
    with respect to $z$.
    By computing the integrals in the latter equation, the function to optimize boils down to
    \[
    z\to 1-2e^{-\frac{z^{2}}{2}}+2e^{-2z^{2}}+2z(2\Phi(2z)-\Phi(z))-2z,
    \]
    where $\Phi(z)$ is the cumulative distribution function of the Gaussian distribution.
    By differentiating the last equation and setting the derivative equal to zero, we retrieve the following equation
    \[
    (2\Phi(2z)-\Phi(z))-1=0.
    \]
    Solving the equation leads us to $z=0.62$, which allows us to retrieve the optimal percentile vector.
\end{proof}

Finally, when the underlying distribution is not symmetric, the optimal percentile vector can still be found by solving the system described in \eqref{eq:system_gen2}.
For the sake of completeness, we retrieve the optimal percentile vector for the Exponential distribution.

\begin{theorem}
\label{thm:exp_3_fac}
When $k=3$, the optimal percentile vector for the Exponential distribution with respect to the Social Cost is $\vec v_{\mathcal{E}}=(0.25,0.67,0.92)$. 
\end{theorem}

\begin{proof}
    Let $y_1$, $y_2$, and $y_3$ denote the support of the projection of $\mu$ over $\PP_3(\erre)$ with respect to the Social Cost.
    Let us denote with $z_1$ and $z_2$ the middle point between $y_1$ and $y_2$, and the middle point of $y_2$ and $y_3$, respectively.
    Following the same argument used in the proof of Theorem \ref{two_facilities_exp}, we have that $y_1$ is the median of $\mu$ restricted to $[0,z_1]$.
    Similarly, we have that $y_2$ is the median of $\mu$ restricted to $[z_1,z_2]$ and $y_3$ is the median of $\mu$ restricted on $[z_3,\infty)$.
    These three conditions allows us to write the following system
    \[
    \begin{cases}
        2(1-e^{-y_1})=1-e^{-z_1}\\
        2((1-e^{-y_1})-(1-e^{-z_1}))=(1-e^{-z_2})-(1-e^{-z_2})\\
        e^{-z_2}=2e^{-y_2}\\
        z_1=\frac{1}{2}(y_1+y_2)\\
        z_2=\frac{1}{2}(y_2+y_3)
    \end{cases}
    \]
    whose solutions are $y_1=\log(\frac{4}{3})$, $y_2=\log(3)$, and $y_3=\log(12)$, which concludes the proof.
\qed
\end{proof}

\end{document}